\documentclass[a4paper,11pt]{article}

\usepackage{amsmath,amssymb,amsthm}
\usepackage{bm}
\usepackage{lineno,hyperref}
\usepackage{graphicx}
\usepackage[top=3cm,bottom=3cm,left=2.5cm,right=2.5cm]{geometry}

\newtheorem{lemma}{Lemma}

\title{
Relaxation heuristics for the set multicover problem\\ with generalized upper bound constraints\thanks{A preliminary version of this paper was presented in \cite{UmetaniS2013}.}}

\author{
Shunji Umetani\thanks{Osaka University, Suita, Osaka, 565-0871, Japan. \texttt{umetani@ist.osaka-u.ac.jp}}, Masanao Arakawa\thanks{Fujitsu Limited, Kawasaki 211-8588, Japan. \texttt{arakawa.masanao@jp.fujitsu.com}}, Mutsunori Yagiura\thanks{Nagoya University, Nagoya 464-8601, Japan. \texttt{yagiura@nagoya-u.jp}}
}

\begin{document}
\maketitle

\begin{abstract}
We consider an extension of the set covering problem (SCP) introducing (i)~multicover and (ii)~generalized upper bound (GUB)~constraints.
For the conventional SCP, the pricing method has been introduced to reduce the size of instances, and several efficient heuristic algorithms based on such reduction techniques have been developed to solve large-scale instances.
However, GUB constraints often make the pricing method less effective, because they often prevent solutions from containing highly evaluated variables together.
To overcome this problem, we develop heuristic algorithms to reduce the size of instances, in which new evaluation schemes of variables are introduced taking account of GUB constraints.
We also develop an efficient implementation of a 2-flip neighborhood local search algorithm that reduces the number of candidates in the neighborhood without sacrificing the solution quality.
In order to guide the search to visit a wide variety of good solutions, we also introduce a path relinking method that generates new solutions by combining two or more solutions obtained so far.
According to computational comparison on benchmark instances, the proposed method succeeds in selecting a small number of promising variables properly and performs quite effectively even for large-scale instances having hard GUB constraints.

\bigskip

\noindent\textbf{keyword:} combinatorial optimization, set covering problem, metaheuristics, local search, Lagrangian relaxation
\end{abstract}

\section{Introduction\label{sec:intro}}
The set covering problem (SCP) is one of representative combinatorial optimization problems.
We are given a set of $m$ elements $i \in M = \{ 1, \dots, m \}$, $n$ subsets $S_j \subseteq M$ ($|S_j| \ge 1$) and their costs $c_j$ ($> 0$) for $j \in N = \{ 1, \dots, n \}$.
We say that $X \subseteq N$ is a cover of $M$ if $\bigcup_{j \in X} S_j = M$ holds.
The goal of SCP is to find a minimum cost cover $X$ of $M$.
The SCP is formulated as a 0-1 integer programming (0-1 IP) problem as follows:
\begin{equation}
\label{eq:scp}
\begin{array}{lll}
\textnormal{minimize} & \displaystyle\sum_{j \in N} c_j x_j & \\
\textnormal{subject to} & \displaystyle\sum_{j \in N} a_{ij} x_j
  \ge 1, & i \in M,\\
& x_j \in \{ 0, 1 \}, & j \in N,
\end{array}
\end{equation}
where $a_{ij} = 1$ if $i \in S_j$ holds and $a_{ij} = 0$ otherwise, and $x_j = 1$ if $j \in X$ and $x_j = 0$ otherwise.
That is, a column $\bm{a}_j = ( a_{1j}, \dots, a_{mj} )^\top$ of matrix $(a_{ij})$ represents the corresponding subset $S_j$ by $S_j = \{ i \in M \mid a_{ij} = 1 \}$.
For notational convenience, for each $i \in M$, let $N_i = \{ j \in N \mid a_{ij} = 1 \}$ be the index set of subsets $S_j$ that contains the element $i$.

The SCP is known to be NP-hard in the strong sense, and there is no polynomial time approximation scheme (PTAS) unless P = NP.
However, the worst-case performance analysis does not necessarily reflect the experimental performance in practice.
The continuous development of mathematical programming has much improved the performance of heuristic algorithms accompanied by advances in computing machinery \cite{CapraraA2000,UmetaniS2007}.
For example, Beasley \cite{BeasleyJE1990a} presented a number of greedy algorithms based on Lagrangian relaxation called the Lagrangian heuristics, and Caprara et~al. \cite{CapraraA1999} introduced pricing techniques into a Lagrangian heuristic algorithm to reduce the size of instances.
Several efficient heuristic algorithms based on Lagrangian heuristics have been developed to solve very large-scale instances with up to 5000 constraints and 1,000,000 variables with deviation within about 1\% from the optimum in a reasonable computing time \cite{CapraraA1999,CasertaM2007,CeriaS1998,YagiuraM2006}.

The SCP has important real applications such as crew scheduling \cite{CapraraA1999}, vehicle routing \cite{HashimotoH2009}, facility location \cite{BorosE2005,FarahaniRZ2012}, and logical analysis of data \cite{BorosE2000}.
However, it is often difficult to formulate problems in real applications as SCP, because they often have additional side constraints in practice.
Most practitioners accordingly formulate them as general mixed integer programming (MIP) problems and apply general purpose solvers, which are usually less efficient compared with solvers specially tailored to SCP.

In this paper, we consider an extension of SCP introducing (i)~multicover and (ii)~generalized upper bound (GUB) constraints, which arise in many real applications of SCP such as vehicle routing \cite{BettinelliA2014,ChoiE2007}, crew scheduling \cite{KohlN2004}, staff scheduling \cite{CapraraA2003,IkegamiA2003} and logical analysis of data \cite{HammerPL2006}.
The multicover constraint is a generalization of covering constraint \cite{PessoaLS2013,VaziraniVV2001}, in which each element $i \in M$ must be covered at least $b_i \in \mathbb{Z}_+$ ($\mathbb{Z}_+$ is the set of nonnegative integers) times.
The GUB constraint is defined as follows.
We are given a partition $\{ G_1, \dots, G_k \}$ of $N$ ($\forall h \not= h^{\prime}$, $G_h \cap G_{h^{\prime}} = \emptyset$, $\bigcup_{h=1}^k G_h = N$).
For each block $G_h \subseteq N$ ($h \in K = \{ 1, \dots, k \}$), the number of selected subsets $S_j$ from the block (i.e., $j \in G_h$) is constrained to be at most $d_h$ ($\le |G_h|$).
We call the resulting problem the \emph{set multicover problem with GUB constraints} (\mbox{SMCP-GUB}), which is formulated as a 0-1 IP problem as follows:
\begin{equation}
\label{eq:smcp-gub}
\begin{array}{llll}
\textnormal{minimize} & \displaystyle z(\bm{x}) = \sum_{j \in N} c_j x_j\\
\textnormal{subject to} & \displaystyle \sum_{j \in N} a_{ij} x_j \ge b_i, & i \in M,\\
 & \displaystyle \sum_{j \in G_h} x_j \le d_h, & h \in K,\\
 & x_j \in \{ 0, 1 \}, & j \in N.\\
\end{array}
\end{equation}

This generalization of SCP substantially extends the variety of its applications.
However, GUB constraints often make the pricing method less effective, because they often prevent solutions from containing highly evaluated variables together.
To overcome this problem, we develop heuristic algorithms to reduce the size of instances, in which new evaluation schemes of variables are introduced taking account of GUB constraints.
We also develop an efficient implementation of a 2-flip neighborhood local search algorithm that reduces the number of candidates in the neighborhood without sacrificing the solution quality.
In order to guide the search to visit a wide variety of good solutions, we also introduce an evolutionary approach called the path relinking method \cite{GloverF1997} that generates new solutions by combining two or more solutions obtained so far.

The \mbox{SMCP-GUB} is NP-hard, and the (supposedly) simpler problem of judging the existence of a feasible solution is NP-complete, since the satisfiability (SAT) problem can be reduced to this decision problem.
We accordingly allow the search to visit infeasible solutions violating multicover constraints and evaluate their quality by the following penalized objective function.
Note that throughout the remainder of the paper, we do not consider solutions that violate the GUB constraints, and the search only visits solutions that satisfy the GUB constraints.
Let $\bm{w} = (w_1, \dots, w_m) \in \mathbb{R}_+^m$ ($\mathbb{R}_+$ is the set of nonnegative real values) be a penalty weight vector.
A solution $\bm{x}$ is evaluated by
\begin{equation}
\label{eq:eval}
\hat{z}(\bm{x},\bm{w}) = \sum_{j \in N} c_j x_j + \sum_{i \in M} w_i \max\left\{ b_i - \sum_{j \in N} a_{ij} x_j, 0 \right\}.
\end{equation}
If the penalty weights $w_i$ are sufficiently large (e.g., $w_i > \sum_{j \in N} c_j$ holds for all $i \in M$), then we can conclude \mbox{SMCP-GUB} to be infeasible when an optimal solution $\bm{x}^{\ast}$ under the penalized objective function $\hat{z}(\bm{x},\bm{w})$ violates at least one multicover constraint.
In our algorithm, the initial penalty weights $\bar{w}_i$ ($i \in M$) are set to $\bar{w}_i = \sum_{j \in N} c_j + 1$ for all $i \in M$.
Starting from the initial penalty weight vector $\bm{w} \leftarrow \bar{\bm{w}}$, the penalty weight vector $\bm{w}$ is adaptively controlled to guide the search to visit better solutions.

We present the outline of the proposed algorithm for \mbox{SMCP-GUB}.
The first set of initial solutions are generated by applying a randomized greedy algorithm several times.
The algorithm then solves a Lagrangian dual problem to obtain a near optimal Lagrangian multiplier vector $\tilde{\bm{u}}$ through a subgradient method (Section~\ref{sec:relax}), which is applied only once in the entire algorithm.
Then, the algorithm applies the following procedures in this order: (i)~heuristic algorithms to reduce the size of instances (Section~\ref{sec:reduction}), (ii)~a 2-flip neighborhood local search algorithm (Section~\ref{sec:local-search}), (iii)~an adaptive control of penalty weights (Section~\ref{sec:weighting}), and (iv)~a path relinking method to generate initial solutions (Section~\ref{sec:path-relink}).
These procedures are iteratively applied until a given time limit has run out.

\section{Lagrangian relaxation and subgradient method\label{sec:relax}}
For a given vector $\bm{u} = (u_1, \dots, u_m) \in \mathbb{R}_+^m$, called a Lagrangian multiplier vector, we consider the following Lagrangian relaxation problem $\textnormal{LR}(\bm{u})$ of \mbox{SMCP-GUB}:
\begin{equation}
\label{eq:relax}
\begin{array}{lll}
\textnormal{minimize} & \multicolumn{2}{l}{z_{\scalebox{0.5}{\textnormal{LR}}}(\bm{u}) = \displaystyle\sum_{j \in N} c_j x_j + \sum_{i \in M} u_i \left( b_i - \sum_{j \in N} a_{ij} x_j \right)}\\
 & \multicolumn{2}{l}{= \displaystyle \sum_{j \in N} \left( c_j - \sum_{i \in M} a_{ij} u_i \right) x_j + \sum_{i \in M} b_i u_i}\\
\textnormal{subject to} & \displaystyle\sum_{j \in G_h} x_j \le d_h, & h \in K,\\
 & x_j \in \{ 0, 1 \}, & j \in N.\\
\end{array}
\end{equation}
We refer to $\tilde{c}_j(\bm{u}) = c_j - \sum_{i \in M} a_{ij} u_i$ as the Lagrangian cost associated with column $j \in N$.
For any $\bm{u} \in \mathbb{R}_+^m$, $z_{\scalebox{0.5}{\textnormal{LR}}}(\bm{u})$ gives a lower bound on the optimal value $z(\bm{x}^{\ast})$ of \mbox{SMCP-GUB} (when it is feasible, i.e., there exists a feasible solution to \mbox{SMCP-GUB}).

The problem of finding a Lagrangian multiplier vector $\bm{u}$ that maximizes $z_{\scalebox{0.5}{\textnormal{LR}}}(\bm{u})$ is called the Lagrangian dual problem ($\textnormal{LRD}$):
\begin{equation}
\label{eq:dual}
\textnormal{maximize} \left\{ z_{\scalebox{0.5}{\textnormal{LR}}}(\bm{u}) \mid \bm{u} \in \mathbb{R}_+^m \right\}.
\end{equation}
For a given $\bm{u} \in \mathbb{R}_+^m$, we can easily compute an optimal solution $\tilde{x}(\bm{u}) = (\tilde{x}_1(\bm{u}), \dots, \tilde{x}_n(\bm{u}))$ to $\textnormal{LR}(\bm{u})$ as follows.
For each block $G_h$ ($h \in K$), if the number of columns $j \in G_h$ satisfying $\tilde{c}_j(\bm{u}) < 0$ is equal to $d_h$ or less, then set $\tilde{x}_j(\bm{u}) \leftarrow 1$ for variables satisfying $\tilde{c}_j(\bm{u}) < 0$ and $\tilde{x}_j(\bm{u}) \leftarrow 0$ for the other variables; otherwise, set $\tilde{x}_j(\bm{u}) \leftarrow 1$ for variables with the $d_h$ lowest Lagrangian costs $\tilde{c}_j(\bm{u})$ and $\tilde{x}_j(\bm{u}) \leftarrow 0$ for the other variables.

The Lagrangian relaxation problem $\textnormal{LR}(\bm{u})$ has integrality property.
That is, an optimal solution to $\textnormal{LR}(\bm{u})$ is also optimal to its linear programming (LP) relaxation problem obtained by replacing $x_j \in \{ 0, 1 \}$ in (\ref{eq:relax}) with $0 \le x_j \le 1$ for all $j \in N$.
In this case, any optimal solution $\bm{u}^{\ast}$ to the dual of the LP relaxation problem of \mbox{SMCP-GUB} is also optimal to LRD, and the optimal value $z_{\scalebox{0.5}{\textnormal{LP}}}$ of the LP relaxation problem of \mbox{SMCP-GUB} is equal to $z_{\scalebox{0.5}{\textnormal{LR}}}(\bm{u}^{\ast})$.

A common approach to compute a near optimal Lagrangian multiplier vector $\tilde{\bm{u}}$ is the subgradient method.
It uses the subgradient vector $\bm{g}(\bm{u}) = (g_1(\bm{u}), \dots, \allowbreak g_m(\bm{u})) \in \mathbb{R}^m$, associated with a given $\bm{u} \in \mathbb{R}_+^m$, defined by
\begin{equation}
\label{eq:subgrad}
g_i(\bm{u}) = b_i - \sum_{j \in N} a_{ij} \tilde{x}_j(\bm{u}).
\end{equation}
This method generates a sequence of nonnegative Lagrangian multiplier vectors $\bm{u}^{(0)}, \bm{u}^{(1)}, \dots$, where $\bm{u}^{(0)}$ is a given initial vector and $\bm{u}^{(l+1)}$ is updated from $\bm{u}^{(l)}$ by the following formula:
\begin{equation}
\label{eq:update-subgrad}
u_i^{(l+1)} \leftarrow \max \left\{ u_i^{(l)} + \lambda \frac{\hat{z}(\bm{x}^{\ast},\bar{\bm{w}}) - z_{\scalebox{0.5}{\textnormal{LR}}}(\bm{u}^{(l)})}{|| \bm{g}(\bm{u}^{(l)}) ||^2} g_i(\bm{u}^{(l)}), 0 \right\}, \quad i \in M,
\end{equation}
where $\bm{x}^{\ast}$ is the best solution obtained so far under the penalized objective function $\hat{z}(\bm{x},\bar{\bm{w}})$ with the initial penalty weight vector $\bar{\bm{w}}$, and $\lambda > 0$ is a parameter called the step size.

When huge instances of SCP are solved, the computing time spent on the subgradient method becomes very large if a naive implementation is used.
Caprara et~al. \cite{CapraraA1999} developed a variant of the pricing method on the subgradient method.
They define a core problem consisting of a small subset of columns $C \subset N$ ($|C| \ll |N|$), chosen among those having low Lagrangian costs $\tilde{c}_j(\bm{u}^{(l)})$ ($j \in N$).
Their algorithm iteratively updates the core problem in a similar fashion that is used for solving large-scale LP problems \cite{BixbyRE1992}.
In order to solve huge instances of \mbox{SMCP-GUB}, we also introduce a pricing method into the basic subgradient method (BSM) described, e.g., in \cite{UmetaniS2007}.

\section{2-flip neighborhood local search\label{sec:local-search}}
The local search (LS) starts from an initial solution $\bm{x}$ and repeats replacing $\bm{x}$ with a better solution $\bm{x}^{\prime}$ in its neighborhood $\mathrm{NB}(\bm{x})$ until no better solution is found in $\mathrm{NB}(\bm{x})$.
For a positive integer $r$, the $r$-flip neighborhood $\mathrm{NB}_r(\bm{x})$ is defined by $\mathrm{NB}_r(\bm{x}) = \{ \bm{x}^{\prime} \in \{ 0,1 \}^n \mid d(\bm{x},\bm{x}^{\prime}) \le r \}$, where $d(\bm{x},\bm{x}^{\prime}) = |\{ j \in N \mid x_j \not= x_j^{\prime} \}|$ is the Hamming distance between $\bm{x}$ and $\bm{x}^{\prime}$.
In other words, $\mathrm{NB}_r(\bm{x})$ is the set of solutions obtainable from $\bm{x}$ by flipping at most $r$ variables.
We first develop a 2-flip neighborhood local search algorithm (2-FNLS) as a basic component of the proposed algorithm.
In order to improve efficiency, 2-FNLS searches $\mathrm{NB}_1(\bm{x})$ first, and then $\mathrm{NB}_2(\bm{x}) \setminus \mathrm{NB}_1(\bm{x})$.

We first describe the algorithm to search $\mathrm{NB}_1(\bm{x})$, called the 1-flip neighborhood search.
Let 
\begin{equation}
\begin{array}{lcl}
\Delta \hat{z}_j^{\uparrow}(\bm{x},\bm{w}) & = & c_j - \displaystyle\sum_{i \in M_L(\bm{x}) \cap S_j} w_i,\\
\Delta \hat{z}_j^{\downarrow}(\bm{x},\bm{w}) & = & - c_j + \displaystyle\sum_{i \in (M_L(\bm{x}) \cup M_E(\bm{x})) \cap S_j} w_i,
\end{array}
\end{equation}
denote the increase of $\hat{z}(\bm{x},\bm{w})$ by flipping $x_j = 0 \to 1$ and $x_j = 1 \to 0$, respectively, where $M_L(\bm{x}) = \{ i \in M \mid \sum_{j \in N} a_{ij} x_j < b_i \}$ and $M_E(\bm{x}) = \{ i \in M \mid \sum_{j \in N} a_{ij} x_j = b_i \}$.
The algorithm first searches for an improved solution obtainable by flipping $x_j = 0 \to 1$ by searching for $j \in N \setminus X(\bm{x})$ satisfying $\Delta \hat{z}_j^{\uparrow}(\bm{x},\bm{w}) < 0$ and $\sum_{j^{\prime} \in G_h} x_{j^{\prime}} < d_h$ for the block $G_h$ containing $j$, where $X(\bm{x}) = \{ j \in N \mid x_j = 1 \}$.
If an improved solution exists, it chooses $j$ with the minimum $\Delta \hat{z}_j^{\uparrow}(\bm{x},\bm{w})$; otherwise, it searches for an improved solution obtainable by flipping $x_j = 1 \to 0$ by searching for $j \in X(\bm{x})$ satisfying $\Delta \hat{z}_j^{\downarrow}(\bm{x},\bm{w}) < 0$.

We next describe the algorithm to search $\mathrm{NB}_2(\bm{x}) \setminus \mathrm{NB}_1(\bm{x})$, called 2-flip neighborhood search.
Yagiura el~al. \cite{YagiuraM2006} developed a 3-flip neighborhood local search algorithm for SCP.
They derived conditions that reduce the number of candidates in $\mathrm{NB}_2(\bm{x}) \setminus \mathrm{NB}_1(\bm{x})$ and $\mathrm{NB}_3(\bm{x}) \setminus \mathrm{NB}_2(\bm{x})$ without sacrificing the solution quality.
However, those conditions are not directly applicable to the 2-flip neighborhood search for \mbox{SMCP-GUB} because of GUB constraints.
Below we derive the following three lemmas that reduce the number of candidates in $\mathrm{NB}_2(\bm{x}) \setminus \mathrm{NB}_1(\bm{x})$ by taking account of GUB constraints.

Let $\Delta \hat{z}_{j_1,j_2}(\bm{x},\bm{w})$ denote the increase of $\hat{z}(\bm{x},\bm{w})$ by flipping the values of $x_{j_1}$ and $x_{j_2}$ simultaneously.
\begin{lemma}
\label{lem:nb1}
Suppose that a solution $\bm{x}$ is locally optimal with respect to $\mathrm{NB}_1(\bm{x})$.
Then $\Delta \hat{z}_{j_1,j_2}(\bm{x},\bm{w}) < 0$ holds, only if $x_{j_1} \not= x_{j_2}$.
\end{lemma}
\begin{proof}
See \ref{app:nb1}. \qed
\end{proof}
This lemma indicates that in searching for improved solutions in $\mathrm{NB}_2(\bm{x}) \setminus \mathrm{NB}_1(\bm{x})$, it is not necessary to consider the simultaneous flip of variables $x_{j_1}$ and $x_{j_2}$ such that $x_{j_1} = x_{j_2} = 0$ or $x_{j_1} = x_{j_2} = 1$.
Based on this, we consider only the set of solutions obtainable by flipping $x_{j_1} = 1 \to 0$ and $x_{j_2} = 0 \to 1$ simultaneously.
We assume that $\sum_{j \in G_h} x_j < d_h$ holds for the block $G_h$ containing $j_2$ or $j_1$ and $j_2$ are in the same block $G_h$, because otherwise the move is infeasible.
Let 
\begin{equation}
\label{eq:nb2}
\Delta \hat{z}_{j_1,j_2}(\bm{x},\bm{w}) = \Delta \hat{z}_{j_1}^{\downarrow}(\bm{x},\bm{w}) + \Delta \hat{z}_{j_2}^{\uparrow}(\bm{x},\bm{w}) - \sum_{i \in M_E(\bm{x}) \cap S_{j_1} \cap S_{j_2}} w_i
\end{equation}
denote the increase of $\hat{z}(\bm{x},\bm{w})$ in this case.
\begin{lemma}
\label{lem:nb2}
Suppose that a solution $\bm{x}$ is locally optimal with respect to $\mathrm{NB}_1(\bm{x})$, $x_{j_1} = 1$ and $x_{j_2} = 0$.
Then $\Delta \hat{z}_{j_1,j_2} (\bm{x},\bm{w}) < 0$ holds, only if at least one of the following two conditions holds.
\begin{description}
\item [\textnormal{(i)}] Both $j_1$ and $j_2$ belong to the same block $G_h$ satisfying $\sum_{j \in G_h} x_j = d_h$.
\item [\textnormal{(ii)}] $M_E(\bm{x}) \cap S_{j_1} \cap S_{j_2} \not= \emptyset$.
\end{description}
\end{lemma}
\begin{proof}
See \ref{app:nb2}. \qed
\end{proof}
\begin{lemma}
\label{lem:nb3}
Suppose that a solution $\bm{x}$ is locally optimal with respect to $\mathrm{NB}_1(\bm{x})$, and for a block $G_h$ and a pair of indices $j_1, j_2 \in G_h$ with $x_{j_1} = 1$ and $x_{j_2} = 0$, $\Delta \hat{z}_{j_1,j_2}(\bm{x},\bm{w}) < 0$, $M_E(\bm{x}) \cap S_{j_1} \cap S_{j_2} = \emptyset$ and $\sum_{j \in G_h} x_j = d_h$ hold.
Let $j_1^{\ast} = \mathrm{arg} \min_{j \in G_h} \Delta \hat{z}_j^{\downarrow}(\bm{x},\bm{w})$ and $j_2^{\ast} = \mathrm{arg} \min_{j \in G_h} \Delta \hat{z}_j^{\uparrow}(\bm{x},\bm{w})$.
Then we have $\Delta \hat{z}_{j_1^{\ast},j_2^{\ast}}(\bm{x},\bm{w}) < 0$.
\end{lemma}
\begin{proof}
See \ref{app:nb3}. \qed
\end{proof}

Note that the condition of Lemma~\ref{lem:nb3} implies that the condition~(i) of Lemma~\ref{lem:nb2} is satisfied.
We can conclude that to find an improved solution satisfying the condition~(i), it suffices to check only one pair for each block $G_h$ satisfying $\sum_{j \in G_h} x_j = d_h$, instead of checking all pairs $(j_1, j_2)$ with $j_1, j_2 \in G_h$, $x_{j_1} = 1$ and $x_{j_2} = 0$ (provided that the algorithm also checks the solutions satisfying the condition~(ii) of Lemma~\ref{lem:nb2}).

The algorithm first searches for an improved solution $\bm{x}^{\prime} \in \mathrm{NB}_2(\bm{x}) \setminus \mathrm{NB}_1(\bm{x})$ satisfying the condition~(i).
For each block $G_h$ ($h \in K$) satisfying $\sum_{j \in G_h} x_j = d_h$, it checks the solution obtained by flipping $x_{j_1} = 1 \to 0$ and $x_{j_2} = 0 \to 1$ for $j_1$ and $j_2$ in $G_h$ with the minimum $\Delta \hat{z}_{j_1}^{\downarrow}(\bm{x},\bm{w})$ and $\Delta \hat{z}_{j_2}^{\uparrow}(\bm{x},\bm{w})$, respectively.
The algorithm then searches for an improved solution $\bm{x}^{\prime} \in \mathrm{NB}_2(\bm{x}) \setminus \mathrm{NB}_1(\bm{x})$ satisfying the condition~(ii).
Let $\mathrm{NB}_2^{(j_1)}(\bm{x})$ denote the subset of $\mathrm{NB}_2(\bm{x})$ obtainable by flipping $x_{j_1} = 1 \to 0$.
The algorithm searches $\mathrm{NB}_2^{(j_1)}(\bm{x})$ for each $j_1 \in X(\bm{x})$ in the ascending order of $\Delta \hat{z}_{j_1}^{\downarrow}(\bm{x},\bm{w})$.
If an improved solution is found, it chooses a pair $j_1$ and $j_2$ with the minimum $\Delta \hat{z}_{j_1,j_2}(\bm{x},\bm{w})$ among those in $\mathrm{NB}_2^{(j_1)}(\bm{x})$.

Algorithm 2-FNLS searches $\mathrm{NB}_1(\bm{x})$ first, and then $\mathrm{NB}_2(\bm{x}) \setminus \mathrm{NB}_1(\bm{x})$.
The algorithm is formally described as follows.

\bigskip

\noindent\underline{{\bfseries Algorithm 2-FNLS}$(\bm{x},\bm{w})$}
\begin{description}
\item [Input:] A solution $\bm{x}$ and a penalty weight vector $\bm{w}$.

\item [Output:] A solution $\bm{x}$.

\item [Step~1:] If $I_1^{\uparrow}(\bm{x}) = \{ j \in N \setminus X(\bm{x}) \mid \Delta \hat{z}_j^{\uparrow}(\bm{x},\bm{w}) < 0, \; \sum_{j^{\prime} \in G_h} x_{j^{\prime}} < d_h \; \textrm{for} \; \textrm{the} \; \textrm{block} \; G_h \allowbreak \; \textrm{containing} \; j \} \not= \emptyset$ holds, choose $j \in I_1^{\uparrow}(\bm{x})$ with the minimum $\Delta \hat{z}_j^{\uparrow}(\bm{x},\bm{w})$, set $x_j \leftarrow 1$ and return to Step~1.

\item [Step~2:] If $I_1^{\downarrow}(\bm{x}) = \{ j \in X(\bm{x}) \mid \Delta \hat{z}_j^{\downarrow}(\bm{x},\bm{w}) < 0 \} \not= \emptyset$ holds, choose $j \in I_1^{\downarrow}(\bm{x})$ with the minimum $\Delta \hat{z}_j^{\downarrow}(\bm{x},\bm{w})$, set $x_j \leftarrow 0$ and return to Step~2.

\item [Step~3:] For each block $G_h$ satisfying $\sum_{j \in G_h} x_j = d_h$ ($h \in K$), if $\Delta \hat{z}_{j_1,j_2}(\bm{x},\bm{w}) < 0$ holds for $j_1$ and $j_2$ with the minimum $\Delta \hat{z}_{j_1}^{\downarrow}(\bm{x},\bm{w})$ and $\Delta \hat{z}_{j_2}^{\uparrow}(\bm{x},\bm{w})$ ($j_1, j_2 \in G_h$), respectively, set $x_{j_1} \leftarrow 0$ and $x_{j_2} \leftarrow 1$.
If the current solution $\bm{x}$ has been updated at least once in Step~3, return to Step~3.

\item [Step~4:] For each $j_1 \in X(\bm{x})$ in the ascending order of $\Delta \hat{z}_{j_1}^{\downarrow}(\bm{x},\bm{w})$, if $I_2(\bm{x}) = \{ j_2 \in N \setminus X(\bm{x}) \mid \Delta \hat{z}_{j_1,j_2}(\bm{x},\bm{w}) < 0, \; ( \sum_{j \in G_h} x_j < d_h \; \textnormal{for} \; \textrm{the} \; \textrm{block} \; G_h \allowbreak \; \textrm{containing} \; j_2 \; \textnormal{or} \; \exists h, \; j_1, j_2 \in G_h), M_E(\bm{x}) \cap S_{j_1} \cap S_{j_2} \not= \emptyset \} \not= \emptyset$ holds, choose $j_2 \in I_2(\bm{x})$ with the minimum $\Delta \hat{z}_{j_1,j_2}(\bm{x},\bm{w})$ and set $x_{j_1} \leftarrow 0$ and $x_{j_2} \leftarrow 1$.
If the current solution $\bm{x}$ has been updated at least once in Step~4, return to Step~1; otherwise output $\bm{x}$ and exit.
\end{description}

We note that 2-FNLS does not necessarily output a locally optimal solution with respect to $\mathrm{NB}_2(\bm{x})$, because the solution $\bm{x}$ is not necessarily locally optimal with respect to $\mathrm{NB}_1(\bm{x})$ in Steps~3 and 4.
Though it is easy to keep the solution $\bm{x}$ locally optimal with respect to $\mathrm{NB}_1(\bm{x})$ in Steps~3 and 4 by returning to Step~1 whenever an improved solution is obtained in Steps~2 or 3, we did not adopt this option because it consumes much computing time just to conclude that the current solution is locally optimal with respect to $\mathrm{NB}_1(\bm{x})$ in most cases.
We also note that the phase to search $\mathrm{NB}_1(\bm{x})$ in the algorithm (i.e., Steps~1 and 2) always finishes with the search for an improved solution obtainable by flipping $x_j = 1 \to 0$ to prevent this phase from stopping at solutions having redundant columns.

Let one-round be the computation needed to find an improved solution in the neighborhood or to conclude that the current solution is locally optimal, including the time to update relevant data structures and/or memory \cite{YagiuraM1999,YagiuraM2001}.
If implemented naively, 2-FNLS requires $\mathrm{O}(\sigma)$ and $\mathrm{O}(n\sigma)$ one-round time for $\mathrm{NB}_1(\bm{x})$ and $\mathrm{NB}_2(\bm{x}) \setminus \mathrm{NB}_1(\bm{x})$, respectively, where $\sigma = \sum_{i \in M} \sum_{j \in N} a_{ij}$.
In order to improve computational efficiency, we keep the following auxiliary data
\begin{equation}
\begin{array}{lcll}
\Delta p_j^{\uparrow}(\bm{x},\bm{w}) &=& \displaystyle\sum_{i \in M_L(\bm{x}) \cap S_j} w_i, & \quad \quad j \in N \setminus X(\bm{x}),\\
\Delta p_j^{\downarrow}(\bm{x},\bm{w}) &=& \displaystyle\sum_{i \in (M_L(\bm{x}) \cup M_E(\bm{x})) \cap S_j} w_i, & \quad \quad j \in X(\bm{x}),
\end{array}
\end{equation}
in memory to compute each $\Delta \hat{z}_j^{\uparrow}(\bm{x},\bm{w}) = c_j - \Delta p_j^{\uparrow}(\bm{x},\bm{w})$ and $\Delta \hat{z}_j^{\downarrow}(\bm{x},\bm{w}) = - c_j + \Delta p_j^{\downarrow}(\bm{x},\bm{w})$ in $\mathrm{O}(1)$ time.
We also keep the values of $s_i(\bm{x}) = \sum_{j \in N} a_{ij} x_j$ ($i \in M$) in memory to update the values of $\Delta p_j^{\uparrow}(\bm{x},\bm{w})$ and $\Delta p_j^{\downarrow}(\bm{x},\bm{w})$ for $j \in N$ in $\mathrm{O}(\tau)$ time when $\bm{x}$ is changed, where $\tau = \max_{j \in N} \sum_{i \in S_j} |N_i|$ (see \ref{app:eval}).

We first consider the one-round time for $\mathrm{NB}_1(\bm{x})$.
In Steps~1 and 2, the algorithm finds $j \in N \setminus X(\bm{x})$ and $j \in X(\bm{x})$ with the minimum $\Delta \hat{z}_j^{\downarrow}(\bm{x},\bm{w})$ and $\Delta \hat{z}_j^{\uparrow}(\bm{x},\bm{w})$ in $\mathrm{O}(n)$ time, respectively, by using the auxiliary data whose update requires $\mathrm{O}(\tau)$ time.
Thus, the one-round time is reduced to $\mathrm{O}(n + \tau)$ for $\mathrm{NB}_1(\bm{x})$.

We next consider the one-round time for $\mathrm{NB}_2(\bm{x}) \setminus \mathrm{NB}_1(\bm{x})$.
In Step~3, the algorithm first finds $j_1$ and $j_2$ with the minimum $\Delta \hat{z}_{j_1}^{\downarrow}(\bm{x},\bm{w})$ and $\Delta \hat{z}_{j_2}^{\uparrow}(\bm{x},\bm{w})$ ($j_1, j_2 \in G_h$), respectively, in $\mathrm{O}(|G_h|)$ time.
The algorithm then evaluates $\Delta \hat{z}_{j_1,j_2}(\bm{x,\bm{w}})$ in $\mathrm{O}(\nu)$ time by using (\ref{eq:nb2}), where $\nu = \max_{j \in N} |S_j|$.
In Step~4, the algorithm first flips $x_{j_1} = 1 \to 0$ and temporarily updates the values of $s_i(\bm{x})$ ($i \in S_{j_1}$), $\Delta p_l^{\uparrow}(\bm{x},\bm{w})$ and $\Delta p_l^{\downarrow}(\bm{x},\bm{w})$ ($l \in N_i$, $i \in S_{j_1}$) in $\mathrm{O}(\tau)$ time so that the memory corresponding to these keeps the values of $s_i(\bm{x}^{\prime})$, $\Delta p_l^{\uparrow}(\bm{x}^{\prime},\bm{w})$ and $\Delta p_l^{\downarrow}(\bm{x}^{\prime},\bm{w})$ for the $\bm{x}^{\prime}$ obtained from $\bm{x}$ by flipping $x_{j_1} = 1 \to 0$.
Then, for searching $\mathrm{NB}_2^{(j_1)}(\bm{x})$, the algorithm evaluates 
\begin{equation}
\begin{array}{lll}
\Delta \hat{z}_{j_1,j_2}(\bm{x},\bm{w}) & = & \Delta \hat{z}_{j_1}^{\downarrow}(\bm{x},\bm{w}) + \Delta \hat{z}_{j_2}^{\uparrow}(\bm{x}^{\prime},\bm{w})\\
& = & \Delta \hat{z}_{j_1}^{\downarrow}(\bm{x},\bm{w}) + c_{j_2} - \Delta p_{j_2}^{\uparrow}(\bm{x}^{\prime},\bm{w})
\end{array}
\end{equation} (in $\mathrm{O}(1)$ time for each pair of $j_1$ and $j_2$) only for each $j_2 \in N \setminus X(\bm{x})$ such that the value of $\Delta p_{j_2}^{\uparrow}(\bm{x},\bm{w})$ has been changed to $\Delta p_{j_2}^{\uparrow}(\bm{x}^{\prime},\bm{w})$ ($\not= \Delta p_{j_2}^{\uparrow}(\bm{x},\bm{w})$) during the temporary update.
Note that the number of such candidates $j_2$ that satisfy $\Delta p_{j_2}^{\uparrow}(\bm{x}^{\prime},\bm{w}) \not= \Delta p_{j_2}^{\uparrow}(\bm{x},\bm{w})$ is $\mathrm{O}(\tau)$.
When an improved solution was not found in $\mathrm{NB}_2^{(j_1)}(\bm{x})$, the updated memory values are restored in $\mathrm{O}(\tau)$ time to the original values $s_i(\bm{x})$, $\Delta p_l^{\uparrow}(\bm{x},\bm{w})$ and $\Delta p_l^{\downarrow}(\bm{x},\bm{w})$ before we try another candidate of $j_1$.
The time to search $\mathrm{NB}_2^{(j_1)}(\bm{x})$ for each $j_1 \in X(\bm{x})$ is therefore $\mathrm{O}(\tau)$.
Thus, the one-round time is reduced to $\mathrm{O}(n + k \nu + n^{\prime} \tau)$ for $\mathrm{NB}_2(\bm{x}) \setminus \mathrm{NB}_1(\bm{x})$, where $n^{\prime} = \sum_{j \in N} x_j = |X(\bm{x})|$.

Because $k \le n$, $\nu \le m$, $\tau \le \sigma$, $n^{\prime} \le n$, $m \le \sigma$, and $n \le \sigma$ always hold, these orders are not worse than those of naive implementation, and they are much better if $\nu \ll m$, $\tau \ll \sigma$ and $n^{\prime} \ll n$ hold, which are the case for many instances.
We also note that the computation time for updating the auxiliary data has little effect on the total computation time of 2-FNLS, because, in most cases, the number of solutions actually visited is much less than that of evaluated neighbor solutions.

\section{Adaptive control of penalty weights\label{sec:weighting}}
We observed that 2-FNLS tends to be attracted to locally optimal solutions of insufficient quality when the penalty weights $w_i$ are large.
We accordingly incorporate a mechanism to adaptively control the values of $w_i$ ($i \in M$) \cite{YagiuraM2006,NonobeK2001,YagiuraM2004}; the algorithm iteratively applies 2-FNLS, updating the penalty weight vector $\bm{w}$ after each call to 2-FNLS.
We call such a sequence of calls to 2-FNLS the \emph{weighting local search} (WLS) according to \cite{SelmanB1993,ThorntonJ2005}.

Let $\bm{x}$ denote the solution at which the previous 2-FNLS stops.
Algorithm WLS resumes 2-FNLS from $\bm{x}$ after updating the penalty weight vector $\bm{w}$.
Starting from an initial penalty weight vector $\bm{w} \leftarrow \bar{\bm{w}}$, where we set $\bar{w}_i = \sum_{j \in N} c_j + 1$ for all $i \in M$, the penalty weight vector $\bm{w}$ is updated as follows.
Let $\bm{x}^{\scalebox{0.5}{\textnormal{best}}}$ denote the best solution obtained in the current call to WLS with respect to the penalized objective function $\hat{z}(\bm{x},\bar{\bm{w}})$ with the initial penalty weight vector $\bar{\bm{w}}$.
If $\hat{z}(\bm{x},\bm{w}) \ge \hat{z}(\bm{x}^{\scalebox{0.5}{\textnormal{best}}},\bar{\bm{w}})$ holds, WLS uniformly decreases the penalty weights $w_i \leftarrow (1 - \eta) w_i$ for all $i \in M$, where the parameter $\eta$ is decided so that for 15\% of variables satisfying $x_j = 1$, the new value of $\Delta \hat{z}_j^{\downarrow}(\bm{x},\bm{w})$ becomes negative.
Otherwise, WLS increases the penalty weights by
\begin{equation}
\label{eq:inc_penalty}
w_i \leftarrow \min \left\{ w_i \left( 1 + \delta \frac{y_i(\bm{x})}{\max_{l \in M} y_l(\bm{x})} \right), \bar{w}_i \right\}, \quad i \in M,
\end{equation}
where $y_i(\bm{x}) = \max\{ b_i - \sum_{j \in N} a_{ij} x_j, 0 \}$ is the amount of violation of the $i$th multicover constraint, and $\delta$ is a parameter that is set to 0.2 in our computational experiments.
Algorithm WLS iteratively applies 2-FNLS, updating the penalty weight vector $\bm{w}$ after each call to 2-FNLS, until the best solution $\bm{x}^{\scalebox{0.5}{\textnormal{best}}}$ with respect to $\hat{z}(\bm{x},\bar{\bm{w}})$ obtained in the current call to WLS has not improved in the last 50 iterations.

\bigskip

\noindent\underline{{\bfseries Algorithm WLS}$(\bm{x})$}
\begin{description}
\item [Input:] A solution $\bm{x}$.

\item [Output:] A solution $\hat{\bm{x}}$ and the best solution $\bm{x}^{\scalebox{0.5}{\textnormal{best}}}$ with respect to $\hat{z}(\bm{x},\bar{\bm{w}})$ obtained in the current call to WLS.

\item [Step~1:] Set $iter \leftarrow 0$, $\bm{x}^{\scalebox{0.5}{\textnormal{best}}} \leftarrow \bm{x}$, $\hat{\bm{x}} \leftarrow \bm{x}$ and $\bm{w} \leftarrow \bar{\bm{w}}$.

\item [Step~2:] Apply $\textnormal{2-FNLS}(\hat{\bm{x}}, \bm{w})$ to obtain an improved solution $\hat{\bm{x}}^{\prime}$ and then set $\hat{\bm{x}} \leftarrow \hat{\bm{x}}^{\prime}$.
Let $\bm{x}^{\prime}$ be the best solution with respect to $\hat{z}(\bm{x},\bar{\bm{w}})$ obtained during the call to $\textnormal{2-FNLS}(\hat{\bm{x}}, \bm{w})$.

\item [Step~3:] If $\hat{z}(\bm{x}^{\prime},\bar{\bm{w}}) < \hat{z}(\bm{x}^{\scalebox{0.5}{\textnormal{best}}},\bar{\bm{w}})$ holds, then set $\bm{x}^{\scalebox{0.5}{\textnormal{best}}} \leftarrow \bm{x}^{\prime}$ and $iter \leftarrow 0$; otherwise, set $iter \leftarrow iter + 1$.
If $iter \ge 50$ holds, output $\hat{\bm{x}}$ and $\bm{x}^{\scalebox{0.5}{\textnormal{best}}}$ and halt.

\item [Step~4:] If $\hat{z}(\hat{\bm{x}},\bm{w}) \ge \hat{z}(\bm{x}^{\scalebox{0.5}{\textnormal{best}}},\bar{\bm{w}})$ holds, then uniformly decrease the penalty weights $w_i$ for all $i \in M$ by $w_i \leftarrow (1 - \eta) w_i$; otherwise, increase the penalty weights $w_i$ for all $i \in M$ by (\ref{eq:inc_penalty}).
Return to Step~2.
\end{description}

\section{Heuristic algorithms to reduce the size of instances\label{sec:reduction}}
For a near optimal Lagrangian multiplier vector $\bm{u}$, the Lagrangian costs $\tilde{c}_j(\bm{u})$ give reliable information on the overall utility of selecting columns $j \in N$ for SCP.
Based on this property, the Lagrangian costs $\tilde{c}_j(\bm{u})$ are often utilized to solve huge instances of SCP.
Similar to the pricing method for solving the Lagrangian dual problem, several heuristic algorithms successively solve a number of subproblems, also called core problems, consisting of a small subset of columns $C \subseteq N$ ($|C| \ll |N|$), chosen among those having low Lagrangian costs $\tilde{c}_j(\bm{u})$ ($j \in C$) \cite{CapraraA1999,CasertaM2007,CeriaS1998,YagiuraM2006}.
The Lagrangian costs $\tilde{c}_j(\bm{u})$ are unfortunately unreliable for selecting columns $j \in N$ for \mbox{SMCP-GUB}, because GUB constraints often prevent solutions from containing more than $d_h$ variables $x_j$ with the lowest Lagrangian costs $\tilde{c}_j(\bm{u})$.
To overcome this problem, we develop two evaluation schemes of columns $j \in N$ for \mbox{SMCP-GUB}.

Before updating the core problem $C$ for every call to WLS, the algorithm heuristically fixes some variables $\hat{x}_j \leftarrow 1$ to reflect the characteristics of the incumbent solution $\bm{x}^{\ast}$ and the current solution $\hat{\bm{x}}$.
Let $\bm{u}$ be a near optimal Lagrangian multiplier vector, and $V = \{ j \in N \mid x_j^{\ast} = \hat{x}_j = 1 \}$ be an index set from which variables to be fixed are chosen.
Let $\tilde{c}_{\max}(\bm{u}) = \max_{j \in V} \tilde{c}_j(\bm{u})$ be the maximum value of the Lagrangian cost $\tilde{c}_j(\bm{u})$ ($j \in V$).
The algorithm randomly chooses a variable $x_j$ ($j \in V$) with probability
\begin{equation}
\label{eq:fix_prob}
\mathrm{prob}_j(\bm{u}) = \frac{\tilde{c}_{\max}(\bm{u}) - \tilde{c}_j(\bm{u})}{\sum_{l \in V} (\tilde{c}_{\max}(\bm{u}) - \tilde{c}_l(\bm{u}))}
\end{equation}
and fixes $\hat{x}_j \leftarrow 1$.
We note that the uniform distribution is used if $\tilde{c}_{\max}(\bm{u}) = \tilde{c}_j(\bm{u})$ holds for all $j \in V$.
The algorithm iteratively chooses and fixes a variable $x_j$ ($j \in V$) until $\sum_{j \in N} a_{ij} x_j \ge b_i$ holds for 20\% of multicover constraints $i \in M$.
It then updates the Lagrangian multiplier $u_i \leftarrow 0$ if $\sum_{j \in F} a_{ij} \ge b_i$ holds for $i \in M$, and computes the Lagrangian costs $\tilde{c}_j(\bm{u})$ for $j \in N \setminus F$, where $F$ is the index set of the fixed variables.
The variable fixing procedure is formally described as follows.

\bigskip

\noindent\underline{{\bfseries Algorithm FIX}$(\bm{x}^{\ast},\hat{\bm{x}},\tilde{\bm{u}})$}
\begin{description}
\item [Input:] The incumbent solution $\bm{x}^{\ast}$, the current solution $\hat{\bm{x}}$ and a near optimal Lagrangian multiplier vector $\tilde{\bm{u}}$.
\item [Output:] A set of fixed variables $F \subset N$ and a Lagrangian multiplier vector $\bm{u}$.
\item [Step~1:] Set $V \leftarrow \{ j \in N \mid x_j^{\ast} = \hat{x}_j = 1 \}$, $F \leftarrow \emptyset$, and $\bm{u} \leftarrow \tilde{\bm{u}}$.
\item [Step~2:] If $|\{ i \in M \mid \sum_{j \in F} a_{ij} \ge b_i \}| \ge 0.2 m$ holds, then set $u_i \leftarrow 0$ for each $i \in M$ satisfying $\sum_{j \in F} a_{ij} \ge b_i$, output $F$ and $\bm{u}$, and halt.
\item [Step~3:] Randomly choose a column $j \in V$ with probability $\mathrm{prob}_j(\bm{u})$ defined by (\ref{eq:fix_prob}), and set $F \leftarrow F \cup \{ j \}$.
Return to Step~2.
\end{description}
Subsequent to the variable fixing procedure, the algorithm updates the instance to be considered by setting $z(\bm{x}) = \sum_{j \in N \setminus F} c_j x_j + \sum_{j \in F} c_j$, $b_i \leftarrow \max\{ b_i - \sum_{j \in F} a_{ij}, 0 \}$ ($i \in M$), and $d_h \leftarrow d_h - |G_h \cap F|$ ($h \in K$).

The first evaluation scheme modifies the Lagrangian costs $\tilde{c}_j(\bm{u})$ to reduce the number of redundant columns $j \in C$ resulting from GUB constraints.
For each block $G_h$ ($h \in K$), let $\theta_h$ be the value of the $(d_h+1)$st lowest Lagrangian cost $\tilde{c}_j(\bm{u})$ among those for columns in $G_h$ if $d_h < |G_h|$ holds and $\theta_h \leftarrow 0$ otherwise.
We then define a score $\rho_j$ for $j \in G_h$, called the normalized Lagrangian score, by $\rho_j = \tilde{c}_j(\bm{u}) - \theta_h$ if $\theta_h < 0$ holds, and $\rho_j = \tilde{c}_j(\bm{u})$ otherwise.

The second evaluation scheme modifies the Lagrangian costs $\tilde{c}_j(\bm{u})$ by replacing the Lagrangian multiplier vector $\bm{u}$ with the adaptively controlled penalty weight vector $\bm{w}$.
We define another score $\phi_j$ for $j \in N$, called the pseudo-Lagrangian score, by $\phi_j = \tilde{c}_j(\bm{w})$.
Intuitive meaning of this score is that we consider a column to be promising if it covers many frequently violated constraints in the recent search.
The variable fixing procedure for the second evaluation scheme is described in a similar fashion to that of the first evaluation scheme by replacing the Lagrangian multiplier vectors $\tilde{\bm{u}}$ and $\bm{u}$ with penalty weight vectors $\tilde{\bm{w}}$ and $\bm{w}$, respectively.

Given a score vector $\bm{\rho}$ (resp., $\bm{\phi}$), a core problem is defined by a subset $C \subset N$ consisting of (i)~columns $j \in N_i$ with the $b_i$ lowest scores $\rho_j$ (resp., $\phi_j$) for each $i \in M$, (ii)~columns $j \in N$ with the $10n^{\prime}$ lowest scores $\rho_j$ (resp., $\phi_j$) (recall that we define $n^{\prime} = \sum_{j \in N} x_j$), and (iii)~columns $j \in X(\bm{x}^{\ast}) \cup X(\hat{\bm{x}})$ for the incumbent solution $\bm{x}^{\ast}$ and the current solution $\hat{\bm{x}}$.
The core problem updating procedure is formally described as follows.

\bigskip

\noindent\underline{{\bfseries Algorithm CORE}$(\bm{\rho}, \bm{x}^{\ast}, \hat{\bm{x}})$}
\begin{description}
\item [Input:] A score vector $\bm{\rho}$, the incumbent solution $\bm{x}^{\ast}$ and the current solution $\hat{\bm{x}}$.
\item [Output:] The core problem $C \subset N$.
\item [Step~1:] For each $i \in M$, let $C_1(i)$ be the set of columns $j \in N_i$ with the $b_i$ lowest $\rho_j$ among those in $N_i$.
Then set $C_1 \leftarrow \bigcup_{i \in M} C_1(i)$.
\item [Step~2:] Set $C_2$ be the set of columns $j \in N$ with the $10n^{\prime}$ lowest $\rho_j$.
\item [Step~3:] Set $C \leftarrow C_1 \cup C_2 \cup X(\bm{x}^{\ast}) \cup X(\hat{\bm{x}})$.
Output $C$ and halt.
\end{description}

\section{Path relinking\label{sec:path-relink}}
The path relinking method \cite{GloverF1997} is an evolutionary approach to integrate intensification and diversification strategies.
This approach generates new solutions by exploring trajectories that connect good solutions.
It starts from one of the good solutions, called an initiating solution, and generates a path by iteratively moving to a solution in the neighborhood that leads toward the other solutions, called guiding solutions.

Because it is preferable to apply path relinking method to solutions of high quality, we keep reference sets $R_1$ and $R_2$ of good solutions with respect to the penalized objective functions $\hat{z}(\bm{x},\bm{w})$ and $\hat{z}(\bm{x},\bar{\bm{w}})$ with the current penalty weight vector $\bm{w}$ and the initial penalty weight vector $\bar{\bm{w}}$, respectively.
Initially $R_1$ and $R_2$ are prepared by repeatedly applying a randomized greedy algorithm, which is the same as Steps~1 and 2 of $\textnormal{2-FNLS}(\bm{0},\bar{\bm{w}})$ except for randomly choosing $j \in I_1^{\uparrow}(\bm{x})$ from those with the five lowest $\Delta \hat{z}_j^{\uparrow}(\bm{x},\bar{\bm{w}})$ in Step~1 (recall that we define $\bar{\bm{w}}$ to be the initial penalty weight vector).
Suppose that the last call to WLS stops at a solution $\hat{\bm{x}}$ and $\bm{x}^{\scalebox{0.5}{\textnormal{best}}}$ is the best solution with respect to $\hat{z}(\bm{x},\bar{\bm{w}})$ obtained during the last call to WLS.
Then, the worst solution $\hat{\bm{x}}^{\scalebox{0.5}{\textnormal{worst}}}$ in $R_1$ (with respect to $\hat{z}(\bm{x},\bm{w})$) is replaced with the solution $\hat{\bm{x}}$ if it satisfies $\hat{z}(\hat{\bm{x}},\bm{w}) \le \hat{z}(\hat{\bm{x}}^{\scalebox{0.5}{\textnormal{worst}}},\bm{w})$ and $\hat{\bm{x}} \not= \bm{x}^{\prime}$ for all $\bm{x}^{\prime} \in R_1$.
The worst solution $\bm{x}^{\scalebox{0.5}{\textnormal{worst}}}$ in $R_2$ (with respect to $\hat{z}(\bm{x},\bar{\bm{w}})$) is replaced with the solution $\bm{x}^{\scalebox{0.5}{\textnormal{best}}}$ if it satisfies $\hat{z}(\bm{x}^{\scalebox{0.5}{\textnormal{best}}},\bar{\bm{w}}) \le \hat{z}(\bm{x}^{\scalebox{0.5}{\textnormal{worst}}},\bar{\bm{w}})$ and $\bm{x}^{\scalebox{0.5}{\textnormal{best}}} \not= \bm{x}^{\prime}$ for all $\bm{x}^{\prime} \in R_2$.

The path relinking method first chooses two solutions $\bm{x}^{\scalebox{0.5}{\textnormal{init}}}$ (initiating solution) and $\bm{x}^{\scalebox{0.5}{\textnormal{guide}}}$ (guiding solution) randomly, one from $R_1$ and another from $R_2$, where we assume that $\hat{z}(\bm{x}^{\scalebox{0.5}{\textnormal{init}}},\bm{w}) \le \hat{z}(\bm{x}^{\scalebox{0.5}{\textnormal{guide}}},\bm{w})$ and $\bm{x}^{\scalebox{0.5}{\textnormal{init}}} \not= \hat{\bm{x}}$ hold.
Let $\xi = d(\bm{x}^{\scalebox{0.5}{\textnormal{init}}}, \bm{x}^{\scalebox{0.5}{\textnormal{guide}}})$ be the Hamming distance between solutions $\bm{x}^{\scalebox{0.5}{\textnormal{init}}}$ and $\bm{x}^{\scalebox{0.5}{\textnormal{guide}}}$.
It then generates a sequence $\bm{x}^{\scalebox{0.5}{\textnormal{init}}} = \bm{x}^{(0)}, \bm{x}^{(1)}, \dots, \bm{x}^{(\xi)} = \bm{x}^{\scalebox{0.5}{\textnormal{guide}}}$ of solutions as follows.
Starting from $\bm{x}^{(0)} \leftarrow \bm{x}^{\scalebox{0.5}{\textnormal{init}}}$, for $l = 0, 1, \dots, \xi-1$, the solution $\bm{x}^{(l+1)}$ is defined to be a solution $\bm{x}^{\prime}$ with the best value of $\hat{z}(\bm{x}^{\prime},\bm{w})$ among those satisfying $\bm{x}^{\prime} \in \mathrm{NB}_1(\bm{x}^{(l)})$ and $d(\bm{x}^{\prime},\bm{x}^{\scalebox{0.5}{\textnormal{guide}}}) < d(\bm{x}^{(l)},\bm{x}^{\scalebox{0.5}{\textnormal{guide}}})$.
The algorithm chooses the first solution $\bm{x}^{(l)}$ ($l = 0,1,\dots,\xi-1$) satisfying $\hat{z}(\bm{x}^{(l)},\bm{w}) \le \hat{z}(\bm{x}^{(l+1)},\bm{w})$ as the next initial solution of WLS.

Given a pair of solutions $\hat{\bm{x}}$ and $\bm{x}^{\scalebox{0.5}{\textnormal{best}}}$ and the current reference sets $R_1$ and $R_2$, the path relinking method outputs the next initial solution $\bm{x}$ of WLS and the updated reference sets $R_1$ and $R_2$.
The path relinking method is formally described as follows.

\bigskip

\noindent\underline{{\bfseries Algorithm PRL}$(\hat{\bm{x}},\bm{x}^{\scalebox{0.5}{\textnormal{best}}},R_1,R_2)$}
\begin{description}
\item [Input:] Solutions $\hat{\bm{x}}$ and $\bm{x}^{\scalebox{0.5}{\textnormal{best}}}$ and reference sets $R_1$ and $R_2$.
\item [Output:] The next initial solution $\bm{x}$ of WLS and the updated reference sets $R_1$ and $R_2$.
\item [Step~1:] Let $\displaystyle\hat{\bm{x}}^{\scalebox{0.5}{\textnormal{worst}}} = \textnormal{arg}\max_{\bm{x}^{\prime} \in R_1} \hat{z}(\bm{x}^{\prime},\bm{w})$ be the worst solution in $R_1$.
If the solution $\hat{\bm{x}}$ satisfies $\displaystyle \hat{z}(\hat{\bm{x}},\bm{w}) \le \hat{z}(\hat{\bm{x}}^{\scalebox{0.5}{\textnormal{worst}}},\bm{w})$ and $\hat{\bm{x}} \not= \bm{x}^{\prime}$ for all $\bm{x}^{\prime} \in R_1$, then set $R_1 \leftarrow R_1 \cup \{ \hat{\bm{x}} \} \setminus \{ \hat{\bm{x}}^{\scalebox{0.5}{\textnormal{worst}}} \}$.
\item [Step~2:] Let $\displaystyle\bm{x}^{\scalebox{0.5}{\textnormal{worst}}} = \textnormal{arg}\max_{\bm{x}^{\prime} \in R_2} \hat{z}(\bm{x}^{\prime},\bar{\bm{w}})$ be the worst solution in $R_2$.
If the solution $\bm{x}^{\scalebox{0.5}{\textnormal{best}}}$ satisfies $\displaystyle \hat{z}(\bm{x}^{\scalebox{0.5}{\textnormal{best}}},\bar{\bm{w}}) \le \hat{z}(\bm{x}^{\scalebox{0.5}{\textnormal{worst}}},\bar{\bm{w}})$ and $\bm{x}^{\scalebox{0.5}{\textnormal{best}}} \not= \bm{x}^{\prime}$ for all $\bm{x}^{\prime} \in R_2$, then set $R_2 \leftarrow R_2 \cup \{ \bm{x}^{\scalebox{0.5}{\textnormal{best}}} \} \setminus \{ \bm{x}^{\scalebox{0.5}{\textnormal{worst}}} \}$.
\item [Step~3:] Randomly choose two solutions $\bm{x}^{\scalebox{0.5}{\textnormal{init}}}$ and $\bm{x}^{\scalebox{0.5}{\textnormal{guide}}}$, one from $R_1$ and another from $R_2$, where we assume that $\hat{z}(\bm{x}^{\scalebox{0.5}{\textnormal{init}}},\bm{w}) \le \hat{z}(\bm{x}^{\scalebox{0.5}{\textnormal{guide}}},\bm{w})$ and $\bm{x}^{\scalebox{0.5}{\textnormal{init}}} \not= \hat{\bm{x}}$ hold.
Set $l \leftarrow 0$ and $\bm{x}^{(l)} \leftarrow \bm{x}^{\scalebox{0.5}{\textnormal{init}}}$.
\item [Step~4:] Set $\bm{x}^{(l+1)} \leftarrow \textnormal{arg}\min \{ \hat{z}(\bm{x}^{\prime},\bm{w}) \mid \bm{x}^{\prime} \in \textnormal{NB}_1(\bm{x}^{(l)}), d(\bm{x}^{\prime},\bm{x}^{\scalebox{0.5}{\textnormal{guide}}}) < d(\bm{x}^{(l)},\bm{x}^{\scalebox{0.5}{\textnormal{guide}}}) \}$.
If $\hat{z}(\bm{x}^{(l)},\bm{w}) > \hat{z}(\bm{x}^{(l+1)},\bm{w})$ holds, set $l \leftarrow l+1$ and return to Step~4; otherwise set $\bm{x} \leftarrow \bm{x}^{(l)}$, output $\bm{x}$, $R_1$ and $R_2$, and halt.
\end{description}

\section{Summary of the proposed algorithm\label{sec:summary}}
We summarize the outline of the proposed algorithm for \mbox{SMCP-GUB} in Figure~\ref{fig:outline}.
The first reference sets $R_1$ and $R_2$ of good solutions are generated by repeating the randomized greedy algorithm (Section~\ref{sec:path-relink}).
The first initial solution $\hat{\bm{x}}$ is set to $\hat{\bm{x}} \leftarrow \mathrm{arg} \min_{\bm{x} \in R_1 \cup R_2} \hat{z}(\bm{x},\bar{\bm{w}})$.
When using the normalized Lagrangian score $\rho_j$, the algorithm obtains a near optimal Lagrangian multiplier vector $\tilde{\bm{u}}$ by the basic subgradient method BSM accompanied by a pricing method \cite{UmetaniS2007} (Section~\ref{sec:relax}), where it is applied only once in the entire algorithm.

The algorithm repeatedly applies the following procedures in this order until a given time has run out.
The heuristic variable fixing algorithm $\textnormal{FIX}(\bm{x}^{\ast},\hat{\bm{x}},\tilde{\bm{u}})$ (Section~\ref{sec:reduction}) decides the index set $F$ of variables to be fixed, and it updates the instance to be considered by fixing variables $\hat{x}_j \leftarrow 1$ ($j \in F$).
The heuristic size reduction algorithm $\textnormal{CORE}(\bm{\rho},\bm{x}^{\ast},\hat{\bm{x}})$ (Section~\ref{sec:reduction}) constructs a core problem $C \subset N$ and fixes variables $\hat{x}_j \leftarrow 0$ ($j \in N \setminus C$).
The weighting local search algorithm $\textnormal{WLS}(\hat{\bm{x}})$ explores good solutions $\hat{\bm{x}}$ and $\bm{x}^{\scalebox{0.5}{\textnormal{best}}}$ with respect to $\hat{z}(\bm{x},\bm{w})$ and $\hat{z}(\bm{x},\bar{\bm{w}})$, respectively, by repeating the 2-flip neighborhood local search algorithm $\textnormal{2-FNLS}(\hat{\bm{x}},\bm{w})$ (Section~\ref{sec:local-search}) while updating the penalty weight vector $\bm{w}$ adaptively (Section~\ref{sec:weighting}), where the initial penalty weights $\bar{w}_i$ ($i \in M$) are set to $\bar{w}_i = \sum_{j \in N} c_j + 1$ for all $i \in M$.
After updating the reference sets $R_1$ and $R_2$, the next initial solution $\hat{\bm{x}}$ is generated by the path relinking method $\textnormal{PRL}(\hat{\bm{x}},\bm{x}^{\scalebox{0.5}{\textnormal{best}}}, R_1, R_2)$ (Section~\ref{sec:path-relink}).

\begin{figure}[ht]
\centering
\includegraphics[scale=0.7]{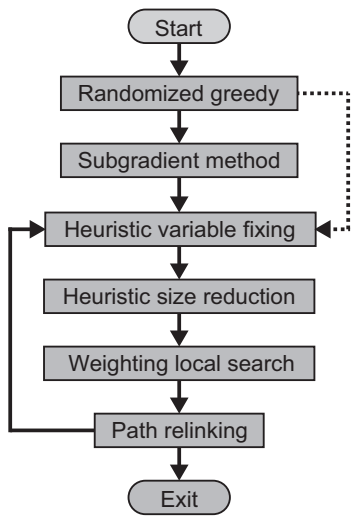}
\caption{Outline of the proposed algorithm for \mbox{SMCP-GUB}\label{fig:outline}}
\end{figure}

\section{Computational results\label{sec:result}}
We first prepared eight classes of random instances for SCP, among which classes G and H were taken from Beasley's OR Library \cite{BeasleyJE1990b} and classes I--N were newly generated by us in a similar manner.
The random instance generator for classes I--N are available at \cite{UmetaniS2015}.
Each class has five instances; we denote instances in class G as G.1, $\dots$, G.5, and other instances in classes H--N similarly.
Another set of benchmark instances is called RAIL arising from a crew pairing problem in an Italian railway company \cite{CapraraA1999,CeriaS1998}.
The summary of these instances are given in Table~\ref{tab:scp}, where the density is defined by $\sum_{i \in M} \sum_{j \in N} a_{ij} / mn$.

For each random instance, we generated four types of \mbox{SMCP-GUB} instances (by the instance generator available at \cite{UmetaniS2015}) with different values of parameters $d_h$ and $|G_h|$ as shown in Table~\ref{tab:smcp-gub}, where all blocks $G_h$ ($h \in K$) have the same size $|G_h|$ and upper bound $d_h$ for each instance.
Here, the right-hand sides of multicover constraints $b_i$ are random integers taken from interval $[1,5]$.
\begin{table}[tb]
\centering
\caption{The benchmark instances for SCP\label{tab:scp}}
\smallskip
{\tabcolsep=0.5em
{\small
\begin{tabular}{lrrrr}\hline
\multicolumn{1}{l}{Instance} & \multicolumn{1}{c}{\#cst.} & \multicolumn{1}{c}{\#vars.} & \multicolumn{1}{c}{Density} & \multicolumn{1}{c}{Cost range} \\ \hline
G.1--G.5 & 1000 & 10,000 & 2.0\% & [1,100] \\
H.1--H.5 & 1000 & 10,000 & 5.0\% & [1,100] \\
I.1--I.5 & 1000 & 50,000 & 1.0\% & [1,100] \\
J.1--J.5 & 1000 & 100,000 & 1.0\% & [1,100] \\
K.1--K.5 & 2000 & 100,000 & 0.5\% & [1,100] \\
L.1--L.5 & 2000 & 200,000 & 0.5\% & [1,100] \\
M.1--M.5 & 5000 & 500,000 & 0.25\% & [1,100] \\
N.1--N.5 & 5000 & 1,000,000 & 0.25\% & [1,100] \\ \hline
RAIL507 & 507 & 63,009 & 1.3\% & [1,2] \\
RAIL516 & 516 & 47,311 & 1.3\% & [1,2] \\
RAIL582 & 582 & 55,515 & 1.2\% & [1,2] \\
RAIL2536 & 2536 & 1,081,841 & 0.4\% & [1,2] \\
RAIL2586 & 2586 & 920,683 & 0.3\% & [1,2] \\
RAIL4284 & 4284 & 1,092,610 & 0.2\% & [1,2] \\
RAIL4872 & 4872 & 968,672 & 0.2\% & [1,2] \\ \hline
\end{tabular}
}
}
\end{table}
\begin{table}[tb]
\centering
\caption{Four types of benchmark instances for \mbox{SMCP-GUB} ($d_h / |G_h|$) \label{tab:smcp-gub}}
\smallskip
{\tabcolsep=0.5em
{\small
\begin{tabular}{lrrrr} \\ \hline
\multicolumn{1}{l}{Instance} & \multicolumn{1}{c}{Type1} & \multicolumn{1}{c}{Type2} & \multicolumn{1}{c}{Type3} & \multicolumn{1}{c}{Type4} \\ \hline
G.1--G.5 & 1/10 & 10/100 & 5/10 & 50/100 \\
H.1--H.5 & 1/10 & 10/100 & 5/50 & 50/100 \\
I.1--I.5 & 1/50 & 10/500 & 5/50 & 50/500 \\
J.1--J.5 & 1/50 & 10/500 & 5/50 & 50/500 \\
K.1--K.5 & 1/50 & 10/500 & 5/50 & 50/500 \\
L.1--L.5 & 1/50 & 10/500 & 5/50 & 50/500 \\
M.1--M.5 & 1/50 & 10/500 & 5/50 & 50/500 \\
N.1--N.5 & 1/100 & 10/1000 & 5/100 & 50/1000 \\ \hline
\end{tabular}
}
}
\end{table}

To the best of our knowledge, there are no specially tailored algorithms for the \mbox{SMCP-GUB}, and \mbox{SMCP-GUB} instances emerging from various applications have been formulated as MIP problems and solved by general purpose solvers in the literature \cite{BettinelliA2014,ChoiE2007,KohlN2004,CapraraA2003,IkegamiA2003,HammerPL2006}.
We accordingly compared the proposed algorithm with two recent MIP solvers called CPLEX12.6 and Gurobi5.6.2 and a local search solver called LocalSolver3.1.
LocalSolver3.1 is not the latest version, but it performs better than more recent version 4.0 for the benchmark instances.
We also compared a 3-flip neighborhood local search algorithm \cite{YagiuraM2006} for SCP instances.
These solvers and the proposed algorithm were tested on a Mac Pro desktop computer with two 2.66~GHz Intel Xeon (six cores) processors and were run on a single thread with time limits shown in Table~\ref{tab:time_limit}.

\begin{table}[tb]
\centering
\caption{Computation time of the tested solvers and the proposed algorithm for \mbox{SMCP-GUB} and SCP\label{tab:time_limit}}
\smallskip
{\tabcolsep=0.5em
{\small
\begin{tabular}{lrr}\hline
\multicolumn{1}{l}{Instance} & \multicolumn{1}{c}{MIP solvers} & \multicolumn{1}{c}{Heuristics} \\ \hline
G.1--G.5 & 3600~s & 600~s \\
H.1--H.5 & 3600~s & 600~s \\
I.1--I.5 & 3600~s & 600~s \\
J.1--J.5 & 3600~s & 600~s \\
K.1--K.5 & 7200~s & 1200~s \\
L.1--L.5 & 7200~s & 1200~s \\
M.1--M.5 & 18,000~s & 3000~s \\
N.1--N.5 & 18,000~s & 3000~s \\ \hline
RAIL507 & 3600~s & 600~s \\
RAIL516 & 3600~s & 600~s \\
RAIL582 & 3600~s & 600~s \\
RAIL2536 & 18,000~s & 3000~s \\
RAIL2586 & 18,000~s & 3000~s \\
RAIL4284 & 18,000~s & 3000~s \\
RAIL4872 & 18,000~s & 3000~s \\ \hline
\end{tabular}
}
}
\end{table}

We first compared the proposed algorithm with different evaluation schemes of variables: (i)~the Lagrangian cost $\tilde{c}_j(\bm{u})$, (ii)~the normalized Lagrangian score $\rho_j$, and (iii)~the pseudo-Lagrangian score $\phi_j$.
We also tested the proposed algorithm without the size reduction mechanism.

Table~\ref{tab:setting-smcp-gub} shows the average objective values of the proposed algorithm with different evaluation schemes of variables for each class of \mbox{SMCP-GUB} instances.
The third column ``$z_{\scalebox{0.5}{\textnormal{LP}}}$'' shows the optimal values of LP relaxation, and the fourth column ``w/o'' shows the proposed algorithm without the size reduction mechanism.
The fifth, sixth, and seventh columns show the results of the proposed algorithm with different evaluation schemes.
The best upper bounds among the compared settings are highlighted in bold.
The last row shows the average relative gaps $\frac{z(\bm{x}) - \bm{z}_{\scalebox{0.5}{\textnormal{best}}}}{z(\bm{x})} \times 100$ (\%) of the compared algorithm, where $\bm{z}_{\scalebox{0.5}{\textnormal{best}}}$ is the best upper bounds among those obtained by all algorithms in this paper.
Table~\ref{tab:core-smcp-gub} shows the average size $\frac{|C|}{|N|} \times 100$ (\%) of the core problem $C$ in the proposed algorithm for each class of \mbox{SMCP-GUB} instances.
The detailed computational results are shown in the online supplement.

We observe that the proposed algorithm with the Lagrangian cost $\tilde{c}_j(\bm{u})$ performs much worse than the algorithm without the size reduction mechanism for types~1 and 2, which indicates that the Lagrangian cost $\tilde{c}_j(\bm{u})$ does not evaluate the promising variables properly for these instances.
The proposed algorithm with the normalized Lagrangian score $\rho_j$ performs much better than the algorithm with the Lagrangian cost $\tilde{c}_j(\bm{u})$ for types~1 and 2, while they show almost the same performance for types~3 and 4.
This is because the normalized Lagrangian score $\rho_j$ takes almost the same value as the Lagrangian cost $\tilde{c}_j(\bm{u})$ for types~3 and 4.
We also observed that the proposed algorithm with the pseudo-Lagrangian score $\phi_j$ performs better than the algorithm with the normalized Lagrangian score $\rho_j$ for most of the tested instances.
These observations indicate that the proposed algorithm with the pseudo-Lagrangian score $\phi_j$ succeeds in selecting a small number of promising variables properly even for \mbox{SMCP-GUB} instances having hard GUB constraints.

\begin{table}[tb]
\centering
\caption{Computational results of the proposed algorithm with different evaluation schemes of variables for \mbox{SMCP-GUB}\label{tab:setting-smcp-gub}}
\smallskip
{\tabcolsep=0.5em
{\small
\begin{tabular}{llrrrrr} \hline
& \multicolumn{1}{l}{Instance} & \multicolumn{1}{c}{$z_{\scalebox{0.5}{\textnormal{LP}}}$} & \multicolumn{1}{c}{w/o} & \multicolumn{1}{c}{Score $\tilde{c}_j(\bm{u})$} & \multicolumn{1}{c}{Score $\rho_j$} & \multicolumn{1}{c}{Score $\phi_j$} \\ \hline
Type1 & G.1--G.5 & 1715.95 & 2357.4 & 2344.8 & 2344.8 & \textbf{2319.4} \\
 & H.1--H.5 & 395.62 & 603.6 & 602.6 & \textbf{601.2} & 604.6 \\
 & I.1--I.5 & 2845.47 & 3913.8 & 5810.2 & 3932.8 & \textbf{3800.4} \\
 & J.1--J.5 & 1466.03 & 2028.4 & 4194.8 & 2011.4 & \textbf{1923.0} \\
 & K.1--K.5 & 5668.91 & 7996.8 & 11873.8 & 8083.8 & \textbf{7691.8} \\
 & L.1--L.5 & 2959.29 & 4254.4 & 8486.8 & 4156.0 & \textbf{3999.4} \\
 & M.1--M.5 & 5476.58 & 8672.0 & 19035.2 & 8343.0 & \textbf{7762.0} \\
 & N.1--N.5 & 4780.91 & 8101.0 & 20613.8 & 7970.8 & \textbf{7010.2} \\ \hline
Type2 & G.1--G.5 & 1508.86 & 1902.4 & \textbf{1900.2} & 1900.6 & 1902.8 \\
 & H.1--H.5 & 367.76 & 508.6 & 504.6 & 502.8 & \textbf{502.6} \\
 & I.1--I.5 & 2702.67 & 3549.0 & 5075.8 & 3550.4 & \textbf{3466.0} \\
 & J.1--J.5 & 1393.72 & 1875.8 & 3471.4 & 1834.6 & \textbf{1791.4} \\
 & K.1--K.5 & 5394.32 & 7432.2 & 11496.0 & 7375.6 & \textbf{7053.6} \\
 & L.1--L.5 & 2820.97 & 3949.4 & 7419.6 & 3830.8 & \textbf{3654.4} \\
 & M.1--M.5 & 5244.54 & 7998.6 & 16781.4 & 7647.8 & \textbf{7187.4} \\
 & N.1--N.5 & 4619.25 & 7694.6 & 19552.4 & 7010.2 & \textbf{6638.6} \\ \hline
Type3 & G.1--G.5 & 708.03 & 765.0 & \textbf{762.2} & \textbf{762.2} & \textbf{762.2} \\
 & H.1--H.5 & 190.16 & 212.0 & \textbf{211.2} & \textbf{211.2} & 211.6 \\
 & I.1--I.5 & 934.96 & 1125.6 & 1115.2 & 1112.8 & \textbf{1103.8} \\
 & J.1--J.5 & 547.93 & 652.4 & 640.8 & 641.4 & \textbf{639.8} \\
 & K.1--K.5 & 1889.22 & 2310.2 & 2274.2 & 2263.2 & \textbf{2248.8} \\
 & L.1--L.5 & 1104.29 & 1339.2 & 1306.8 & 1312.0 & \textbf{1300.6} \\
 & M.1--M.5 & 2083.85 & 2708.2 & 2572.8 & 2563.8 & \textbf{2542.6} \\
 & N.1--N.5 & 1747.88 & 2441.4 & 2347.8 & 2308.6 & \textbf{2225.8} \\ \hline
Type4 & G.1--G.5 & 691.14 & 732.8 & 736.4 & \textbf{730.2} & 730.8 \\
 & H.1--H.5 & 187.71 & 204.4 & 203.6 & 203.6 & \textbf{203.2} \\
 & I.1--I.5 & 917.48 & 1074.8 & 1060.8 & 1064.2 & \textbf{1059.4} \\
 & J.1--J.5 & 539.06 & 625.0 & 615.4 & 615.4 & \textbf{611.2} \\
 & K.1--K.5 & 1853.17 & 2217.0 & 2166.8 & 2174.0 & \textbf{2161.8} \\
 & L.1--L.5 & 1088.47 & 1286.0 & 1257.4 & 1257.8 & \textbf{1250.8} \\
 & M.1--M.5 & 2051.53 & 2582.6 & 2484.0 & 2492.4 & \textbf{2453.8} \\
 & N.1--N.5 & 1724.53 & 2337.6 & 2240.6 & 2233.4 & \textbf{2167.4} \\ \hline
Avg. gap & & & 4.69\% & 19.13\% & 2.84\% & \textbf{0.56\%} \\ \hline
\end{tabular}
}
}
\end{table}

\begin{table}[tb]
\centering
\caption{The size of the core problem in the proposed algorithm with different evaluation schemes of variables for \mbox{SMCP-GUB}\label{tab:core-smcp-gub}}
\smallskip
{\tabcolsep=0.5em
{\small
\begin{tabular}{llrrr} \hline
& \multicolumn{1}{l}{Instance} & \multicolumn{1}{c}{Score $\tilde{c}_j(\bm{u})$} & \multicolumn{1}{c}{Score $\rho_j$} & \multicolumn{1}{c}{Score $\phi_j$} \\ \hline
Type1 & G.1--G.5 & 20.51\% & 20.37\% & 24.84\% \\
 & H.1--H.5 & 10.14\% & 10.12\% & 12.53\% \\
 & I.1--I.5 & 5.75\% & 5.41\% & 6.67\% \\
 & J.1--J.5 & 2.79\% & 2.68\% & 3.30\% \\
 & K.1--K.5 & 5.79\% & 5.49\% & 6.65\% \\
 & L.1--L.5 & 2.86\% & 2.73\% & 3.34\% \\
 & M.1--M.5 & 2.49\% & 2.45\% & 2.92\% \\
 & N.1--N.5 & 1.19\% & 1.19\% & 1.40\% \\ \hline
Type2 & G.1--G.5 & 18.41\% & 18.45\% & 21.90\% \\
 & H.1--H.5 & 9.32\% & 9.26\% & 11.11\% \\
 & I.1--I.5 & 5.52\% & 5.17\% & 6.30\% \\
 & J.1--J.5 & 2.74\% & 2.58\% & 3.16\% \\
 & K.1--K.5 & 5.53\% & 5.26\% & 6.31\% \\
 & L.1--L.5 & 2.76\% & 2.64\% & 3.16\% \\
 & M.1--M.5 & 2.41\% & 2.35\% & 2.75\% \\
 & N.1--N.5 & 1.15\% & 1.12\% & 1.33\% \\ \hline
Type3 & G.1--G.5 & 23.48\% & 23.48\% & 25.75\% \\
 & H.1--H.5 & 11.65\% & 11.65\% & 12.72\% \\
 & I.1--I.5 & 6.07\% & 6.05\% & 7.18\% \\
 & J.1--J.5 & 3.04\% & 3.04\% & 3.56\% \\
 & K.1--K.5 & 6.14\% & 6.12\% & 7.18\% \\
 & L.1--L.5 & 3.06\% & 3.07\% & 3.57\% \\
 & M.1--M.5 & 2.65\% & 2.65\% & 3.07\% \\
 & N.1--N.5 & 1.26\% & 1.25\% & 1.45\% \\ \hline
Type4 & G.1--G.5 & 23.23\% & 23.22\% & 25.18\% \\
 & H.1--H.5 & 11.55\% & 11.55\% & 12.29\% \\
 & I.1--I.5 & 5.95\% & 5.97\% & 7.00\% \\
 & J.1--J.5 & 3.03\% & 3.03\% & 3.54\% \\
 & K.1--K.5 & 6.04\% & 6.03\% & 7.07\% \\
 & L.1--L.5 & 3.02\% & 3.03\% & 3.51\% \\
 & M.1--M.5 & 2.62\% & 2.63\% & 3.01\% \\
 & N.1--N.5 & 1.24\% & 1.23\% & 1.40\% \\ \hline
\end{tabular}
}
}
\end{table}

Table~\ref{tab:setting-scp} shows the average objective values of the proposed algorithm for each class of SCP instances, in which we omit the results for the normalized Lagrangian score $\rho_j$ because it takes exactly the same value as the Lagrangian cost $\tilde{c}_j(\bm{u})$ for SCP instances.
Table~\ref{tab:core-scp} shows the average size $\frac{|C|}{|N|} \times 100$ (\%) of the core problem $C$ in the proposed algorithm for each class of SCP instances.
The detailed computational results are shown in the online supplement.

We observe that the proposed algorithm with the Lagrangian cost $\tilde{c}_j(\bm{u})$ and the pseudo-Lagrangian score $\phi_j$ performs better than the algorithm without the size reduction mechanism, and the algorithm with the Lagrangian cost $\tilde{c}_j(\bm{u})$ performs best for the RAIL instances.
These observations indicate that the proposed algorithm with the pseudo-Lagrangian score $\phi_j$ succeeds in selecting a small number of promising variables properly for SCP instances.

\begin{table}[tb]
\centering
\caption{Computational results of the proposed algorithm with different evaluation schemes of variables for SCP\label{tab:setting-scp}}
\smallskip
{\tabcolsep=0.5em
{\small
\begin{tabular}{lrrrr} \hline
\multicolumn{1}{l}{Instance} & \multicolumn{1}{c}{$z_{\scalebox{0.5}{\textnormal{LP}}}$} & \multicolumn{1}{c}{w/o} & \multicolumn{1}{c}{Score $\tilde{c}_j(\bm{u})$} & \multicolumn{1}{c}{Score $\phi_j$}\\ \hline
G.1--G.5 & 149.48 & \textbf{166.4} & \textbf{166.4} & \textbf{166.4} \\
H.1--H.5 & 45.67 & \textbf{59.6} & \textbf{59.6} & \textbf{59.6} \\
I.1--I.5 & 138.96 & 160.2 & 159.0 & \textbf{158.8} \\
J.1--J.5 & 104.78 & 132.6 & 130.8 & \textbf{130.6} \\
K.1--K.5 & 276.66 & 321.8 & 319.2 & \textbf{318.2} \\
L.1--L.5 & 209.33 & 268.2 & \textbf{263.0} & 264.0 \\
M.1--M.5 & 415.77 & 575.6 & \textbf{565.2} & 565.8 \\
N.1--N.5 & 348.92 & 534.2 & \textbf{516.4} & 518.2 \\ \hline
RAIL507 & 172.14 & 179 & \textbf{175} & 178 \\
RAIL516 & 182.00 & 183 & \textbf{182} & 183 \\
RAIL582 & 209.71 & 217 & \textbf{212} & 216 \\
RAIL2536 & 688.39 & 717 & \textbf{693} & 715 \\
RAIL2586 & 935.92 & 1010 & \textbf{965} & 988 \\
RAIL4284 & 1054.05 & 1121 & \textbf{1080} & 1127 \\
RAIL4872 & 1509.63 & 1618 & \textbf{1564} & 1591 \\ \hline
Avg. gap & & 2.74\% & \textbf{1.28\%} & 1.63\% \\ \hline
\end{tabular}
}
}
\end{table}

\begin{table}[tb]
\centering
\caption{The size of the core problem in the proposed algorithm with different evaluation schemes of variables for SCP\label{tab:core-scp}}
\smallskip
{\tabcolsep=0.5em
{\small
\begin{tabular}{lrr} \hline
\multicolumn{1}{l}{Instance} & \multicolumn{1}{c}{Score $\tilde{c}_j(\bm{u})$} & \multicolumn{1}{c}{Score $\phi_j$}\\ \hline
G.1--G.5 & 10.36\% & 10.36\% \\
H.1--H.5 & 5.17\% & 5.25\% \\
I.1--I.5 & 2.93\% & 2.93\% \\
J.1--J.5 & 1.31\% & 1.32\% \\
K.1--K.5 & 2.96\% & 2.94\% \\
L.1--L.5 & 1.31\% & 1.33\% \\
M.1--M.5 & 1.12\% & 1.14\% \\
N.1--N.5 & 0.52\% & 0.53\% \\ \hline
RAIL507 & 1.98\% & 2.18\% \\
RAIL516 & 3.46\% & 3.84\% \\
RAIL582 & 3.03\% & 3.25\% \\
RAIL2536 & 0.43\% & 0.50\% \\
RAIL2586 & 0.71\% & 0.82\% \\
RAIL4284 & 0.71\% & 0.81\% \\
RAIL4872 & 1.16\% & 1.34\% \\ \hline
\end{tabular}
}
}
\end{table}

We next compared variations of the proposed algorithm obtained by applying one of the following three modifications: (i)~replace the 2-flip neighborhood search with the 1-flip neighborhood local search algorithm (i.e., apply only Steps~1 and 2 in Algorithm 2-FNLS), (ii)~exclude the path relinking method (i.e., exclude Algorithm PRL from the proposed algorithm), and (iii)~replace the randomized greedy algorithm with the uniformly random selection of $j \in I_1^{\uparrow}(\bm{x})$, where we tested all of these three variations with the pseudo-Lagrangian score $\phi_j$.
Tables~\ref{tab:variation-smcp-gub} and \ref{tab:variation-scp} show the average objective values of the three variations of the proposed algorithm for \mbox{SMCP-GUB} and SCP instances.
The columns ``1-FNLS'', ``No-PRL'' and ``No-GR'' show the results of the proposed algorithm with the above modifications (i), (ii) and (iii), respectively.
The last column ``Proposed'' shows the results of the proposed algorithm (without such modifications).
The detailed computational results are shown in the online supplement.

Comparing columns ``1-FNLS'' with ``Proposed'', we observed that the 2-flip neighborhood local search algorithm performs much better than the 1-flip neighborhood local search algorithm for \mbox{SMCP-GUB} instances, while they show almost the same performance for SCP instances.
The proposed algorithm performs better on average than that obtained by excluding the path relinking and that obtained by replacing the randomized greedy algorithm with the uniformly random selection of $j \in I_1^{\uparrow}(\bm{x})$ for \mbox{SMCP-GUB} and SCP instances, where the performance depends on types of instances and the differences are small.
These observations indicate that procedures for generating initial solutions have less influence on its performance than those in the evaluation schemes of variables and the neighborhood search procedures.

\begin{table}[tb]
\centering
\caption{Computational results of variations of the proposed algorithm for \mbox{SMCP-GUB}\label{tab:variation-smcp-gub}}
\smallskip
{\tabcolsep=0.5em
{\small
\begin{tabular}{llrrrrr} \hline
& \multicolumn{1}{l}{Instance} & \multicolumn{1}{c}{$z_{\scalebox{0.5}{\textnormal{LP}}}$} & \multicolumn{1}{c}{\mbox{1-FNLS}} & \multicolumn{1}{c}{\mbox{No-PRL}} & \multicolumn{1}{c}{\mbox{No-GR}} & \multicolumn{1}{c}{Proposed} \\ \hline
Type1 & G.1--G.5 & 1715.95 & 2354.2 & 2330.2 & 2342.6 & \textbf{2319.4} \\
 & H.1--H.5 & 395.62 & 608.8 & \textbf{601.6} & 604.0 & 604.6 \\
 & I.1--I.5 & 2845.47 & 3850.2 & 3784.2  & \textbf{3778.8} & 3800.4 \\
 & J.1--J.5 & 1466.03 & 1973.8 & 1960.0 & 1937.0 & \textbf{1923.0} \\
 & K.1--K.5 & 5668.91 & 7744.2 & 7692.0 & 7713.0 & \textbf{7691.8} \\
 & L.1--L.5 & 2959.29 & 4063.6 & \textbf{3982.0} & 3988.0 & 3999.4 \\
 & M.1--M.5 & 5476.58 & 7949.6 & \textbf{7685.6} & 7756.4 & 7762.0 \\
 & N.1--N.5 & 4780.91 & 7312.4 & \textbf{6914.2} & 7065.2 & 7010.2 \\ \hline
Type2 & G.1--G.5 & 1508.86 & 1894.2 & 1913.2 & \textbf{1885.8} & 1902.8 \\
 & H.1--H.5 & 367.76 & 506.4 & 503.0 & \textbf{499.8} & 502.6 \\
 & I.1--I.5 & 2702.67 & 3504.8 & 3479.8 & \textbf{3452.6} & 3466.0 \\
 & J.1--J.5 & 1393.72 & 1831.2 & 1793.6 & \textbf{1782.8} & 1791.4 \\
 & K.1--K.5 & 5394.32 & 7262.2 & 7062.6 & \textbf{7044.0} & 7053.6 \\
 & L.1--L.5 & 2820.97 & 3798.0 & 3654.8 & 3659.2 & \textbf{3654.4} \\
 & M.1--M.5 & 5244.54 & 7400.6 & \textbf{7139.4} & 7219.4 & 7187.4 \\
 & N.1--N.5 & 4619.25 & 6893.2 & \textbf{6538.4} & 6635.4 & 6638.6 \\ \hline
Type3 & G.1--G.5 & 708.03 & \textbf{761.0} & 764.4 & 763.6 & 762.2 \\
 & H.1--H.5 & 190.16 & 211.8 & \textbf{211.4} & \textbf{211.4} & 211.6 \\
 & I.1--I.5 & 934.96 & 1107.4 & 1112.0 & 1104.0 & \textbf{1103.8} \\
 & J.1--J.5 & 547.93 & 646.4 & 641.2 & 640.6 & \textbf{639.8} \\
 & K.1--K.5 & 1889.22 & 2280.8 & 2262.4 & 2255.0 & \textbf{2248.8} \\
 & L.1--L.5 & 1104.29 & 1314.0 & 1306.6 & 1304.8 & \textbf{1300.6} \\
 & M.1--M.5 & 2083.85 & 2572.2 & \textbf{2541.8} & 2546.0 & 2542.6 \\
 & N.1--N.5 & 1747.88 & 2272.8 & 2229.6 & 2241.8 & \textbf{2225.8} \\ \hline
Type4 & G.1--G.5 & 691.14 & \textbf{730.8} & 732.4 & \textbf{730.8} & \textbf{730.8} \\
 & H.1--H.5 & 187.71 & 204.6 & 203.8 & 204.2 & \textbf{203.2} \\
 & I.1--I.5 & 917.48 & 1064.6 & 1059.4 & \textbf{1057.2} & 1059.4 \\
 & J.1--J.5 & 539.06 & 615.8 & 616.6 & 612.2 & \textbf{611.2} \\
 & K.1--K.5 & 1853.17 & 2188.0 & 2162.0 & \textbf{2160.4} & 2161.8 \\
 & L.1--L.5 & 1088.47 & 1271.4 & 1257.2 & 1253.8 & \textbf{1250.8} \\
 & M.1--M.5 & 2051.53 & 2485.0 & \textbf{2447.6} & 2457.0 & 2453.8 \\
 & N.1--N.5 & 1724.53 & 2202.4 & \textbf{2147.2} & 2178.0 & 2167.4 \\ \hline
Avg. gap & & & 1.98\% & 0.58\% & 0.64\% & \textbf{0.56\%} \\ \hline
\end{tabular}
}
}
\end{table}

\begin{table}[tb]
\centering
\caption{Computational results of variations of the proposed algorithm for SCP\label{tab:variation-scp}}
\smallskip
{\tabcolsep=0.5em
{\small
\begin{tabular}{lrrrrr} \hline
\multicolumn{1}{l}{Instance} & \multicolumn{1}{c}{$z_{\scalebox{0.5}{\textnormal{LP}}}$} & \multicolumn{1}{c}{\mbox{1-FNLS}} & \multicolumn{1}{c}{\mbox{No-PRL}} & \multicolumn{1}{c}{\mbox{No-GR}} & \multicolumn{1}{c}{Proposed} \\ \hline
G.1--G.5 & 149.48 & \textbf{166.4} & 166.6 & 166.6 & \textbf{166.4} \\
H.1--H.5 & 45.67 & \textbf{59.6} & \textbf{59.6} & \textbf{59.6} & \textbf{59.6} \\
I.1--I.5 & 138.96 & 159.0 & 159.2 & \textbf{158.8} & \textbf{158.8} \\
J.1--J.5 & 104.78 & 130.4 & 130.8 & \textbf{130.0} & 130.6 \\
K.1--K.5 & 276.66 & 318.6 & 320.0 & \textbf{318.2} & \textbf{318.2} \\
L.1--L.5 & 209.33 & \textbf{263.4} & 264.4 & 264.2 & 264.0 \\
M.1--M.5 & 415.77 & \textbf{564.4} & 567.0 & 565.8 & 565.8 \\
N.1--N.5 & 348.92 & 520.6 & \textbf{517.0} & 520.4 & 518.2 \\ \hline
RAIL507 & 172.14 & 179 & \textbf{176} & 181 & 178 \\
RAIL516 & 182.00 & \textbf{183} & \textbf{183} & \textbf{183} & \textbf{183} \\
RAIL582 & 209.71 & 218 & \textbf{212} & 214 & 216 \\
RAIL2536 & 688.39 & 728 & \textbf{711} & 712 & 715 \\
RAIL2586 & 935.92 & 1000 & \textbf{979} & 991 & 988 \\
RAIL4284 & 1054.05 & 1135 & \textbf{1109} & 1135 & 1127 \\
RAIL4872 & 1509.63 & 1611 & \textbf{1579} & 1599 & 1591 \\ \hline
Avg. gap & & 1.77\% & \textbf{1.62\%} & 1.68\% & 1.63\% \\ \hline
\end{tabular}
}
}
\end{table}

We finally compared the proposed algorithm with the above-mentioned recent solvers, where we tested the proposed algorithm with the pseudo-Lagrangian score $\phi_j$.
Tables~\ref{tab:solver-smcp-gub} and \ref{tab:solver-scp} show the average objective values of the compared algorithms for each class of \mbox{SMCP-GUB} and SCP instances, respectively, where the results with asterisks ``$\ast$'' indicate that the obtained feasible solutions were proven to be optimal.
The detailed computational results are shown in the online supplement.

We observe that the proposed algorithm performs better than CPLEX12.6, Gurobi5.6.2 and LocalSolver3.1 for all types of \mbox{SMCP-GUB} instances and classes G--N of SCP instances while it performs worse than CPLEX12.6 and Gurobi5.6.2 for class RAIL of SCP instances.
We also observe that the 3-flip neighborhood local search algorithm \cite{YagiuraM2006} achieves best upper bounds for almost all classes of SCP instances.
These observations indicate that the variable fixing and pricing techniques based on the LP and/or Lagrangian relaxation are greatly affected by the gaps between lower and upper bounds, and they may not work effectively for the instances having large gaps.
For such instances, the proposed algorithm with the pseudo-Lagrangian score $\phi_j$ succeeds in evaluating the promising variables properly.

\begin{table}[tb]
\centering
\caption{Computational results of the tested solvers and the proposed algorithm for \mbox{SMCP-GUB}\label{tab:solver-smcp-gub}}
\smallskip
{\tabcolsep=0.5em
{\small
\begin{tabular}{llrrrrrr} \hline
& \multicolumn{1}{l}{Instance} & \multicolumn{1}{c}{$z_{\scalebox{0.5}{\textnormal{LP}}}$} & \multicolumn{1}{c}{CPLEX} & \multicolumn{1}{c}{Gurobi} & \multicolumn{1}{c}{LocalSolver} & \multicolumn{1}{c}{Proposed} \\ \hline
Type1 & G.1--G.5 & 1715.95 & 2590.8 & 2629.6 & 4096.2 & \textbf{2319.4} \\
 & H.1--H.5 & 395.62 & 680.8 & 705.4 & 1217.2 & \textbf{604.6} \\
 & I.1--I.5 & 2845.47 & 4757.2 & 4544.6 & 7850.8 & \textbf{3800.4} \\
 & J.1--J.5 & 1466.03 & 2656.4 & 2329.6 & 4744.6 & \textbf{1923.0} \\
 & K.1--K.5 & 5668.91 & 12363.6 & 9284.4 & 16255.0 & \textbf{7691.8} \\
 & L.1--L.5 & 2959.29 & 7607.4 & 4749.2 & 9940.6 & \textbf{3999.4} \\
 & M.1--M.5 & 5476.58 & 16211.0 & 10152.2 & 21564.8 & \textbf{7762.0} \\
 & N.1--N.5 & 4780.91 & 16970.8 & 8966.2 & 21618.8 & \textbf{7010.2} \\ \hline
Type2 & G.1--G.5 & 1508.86 & 2053.8 & 2023.8 & 3327.0 & \textbf{1902.8} \\
 & H.1--H.5 & 367.76 & 555.2 & 560.0 & 915.4 & \textbf{502.6} \\
 & I.1--I.5 & 2702.67 & 4382.0 & 3939.6 & 6514.4 & \textbf{3466.0} \\
 & J.1--J.5 & 1393.72 & 3156.2 & 2132.6 & 3837.6 & \textbf{1791.4} \\
 & K.1--K.5 & 5394.32 & 15763.6 & 7952.2 & 13647.0 & \textbf{7053.6} \\
 & L.1--L.5 & 2820.97 & 15683.0 & 4338.4 & 8117.4 & \textbf{3654.4} \\
 & M.1--M.5 & 5244.54 & 36794.8 & 8600.8 & 17614.4 & \textbf{7187.4} \\
 & N.1--N.5 & 4619.25 & 36970.0 & 8154.4 & 17732.6 & \textbf{6638.6} \\ \hline
Type3 & G.1--G.5 & 708.03 & 771.0 & 769.8 & 1585.4 & \textbf{762.2} \\
 & H.1--H.5 & 190.16 & 214.6 & 216.4 & 554.8 & \textbf{211.6} \\
 & I.1--I.5 & 934.96 & 1223.6 & 1176.2 & 2326.4 & \textbf{1103.8} \\
 & J.1--J.5 & 547.93 & 707.4 & 669.2 & 1802.8 & \textbf{639.8} \\
 & K.1--K.5 & 1889.22 & 3802.8 & 2530.6 & 4788.8 & \textbf{2248.8} \\
 & L.1--L.5 & 1104.29 & 3318.2 & 1384.6 & 4294.4 & \textbf{1300.6} \\
 & M.1--M.5 & 2083.85 & 8404.4 & 2691.6 & 9575.8 & \textbf{2542.6} \\
 & N.1--N.5 & 1747.88 & 8394.0 & 2597.6 & 11805.4 & \textbf{2225.8} \\ \hline
Type4 & G.1--G.5 & 691.14 & 735.2 & 731.8 & 1495.0 & \textbf{730.8} \\
 & H.1--H.5 & 187.71 & 206.8 & 207.2 & 558.8 & \textbf{203.2} \\
 & I.1--I.5 & 917.48 & 1150.0 & 1089.6 & 2109.0 & \textbf{1059.4} \\
 & J.1--J.5 & 539.06 & 644.4 & 629.8 & 1698.4 & \textbf{611.2} \\
 & K.1--K.5 & 1853.17 & 2734.2 & 2231.6 & 4654.4 & \textbf{2161.8} \\
 & L.1--L.5 & 1088.47 & 1566.0 & 1299.6 & 3809.6 & \textbf{1250.8} \\
 & M.1--M.5 & 2051.53 & 7572.6 & 2547.2 & 9909.4 & \textbf{2453.8} \\
 & N.1--N.5 & 1724.53 & 8371.8 & 2339.8 & 9906.8 & \textbf{2167.4} \\ \hline
Avg. gap & & & 34.12\% & 10.44\% & 58.37\% & \textbf{0.56\%} \\ \hline
\end{tabular}
}
}
\end{table}
\begin{table}[tb]
\centering
\caption{Computational results of the tested solvers and the proposed algorithm for SCP\label{tab:solver-scp}}
\smallskip
{\tabcolsep=0.5em
{\small
\begin{tabular}{lrrrrrrr} \hline
\multicolumn{1}{l}{Instance} & \multicolumn{1}{c}{$z_{\scalebox{0.5}{\textnormal{LP}}}$} & \multicolumn{1}{c}{CPLEX} & \multicolumn{1}{c}{Gurobi} & \multicolumn{1}{c}{LocalSolver} & \multicolumn{1}{c}{Yagiura et~al.} & \multicolumn{1}{c}{Proposed} \\ \hline
G.1--G.5 & 149.48 & 166.6 & 166.6 & 316.6 & \textbf{166.4} & \textbf{166.4}\\
H.1--H.5 & 45.67 & 60.2 & 60.2 & 169.2 & \textbf{59.6} & \textbf{59.6} \\
I.1--I.5 & 138.96 & 162.4 & 162.4 & 403.4 & \textbf{158.0} & 158.8 \\
J.1--J.5 & 104.78 & 137.8 & 135.2 & 335.8 & \textbf{129.0} & 130.6 \\
K.1--K.5 & 276.66 & 325.8 & 325.2 & 828.2 & \textbf{313.2} & 318.2 \\
L.1--L.5 & 209.33 & 285.4 & 273.6 & 918.6 & \textbf{258.6} & 264.0 \\
M.1--M.5 & 415.77 & 637.2 & 609.8 & 2163.8 & \textbf{550.2} & 565.8 \\
N.1--N.5 & 348.92 & 753.0 & 577.4 & 2382.2 & \textbf{503.8} & 518.2 \\ \hline
RAIL507 & 172.14 & $\ast$\textbf{174} & $\ast$\textbf{174} & 187 & \textbf{174} & 178 \\
RAIL516 & 182.00 & $\ast$\textbf{182} & $\ast$\textbf{182} & 189 & \textbf{182} & 183 \\
RAIL582 & 209.71 & $\ast$\textbf{211} & $\ast$\textbf{211} & 227 & \textbf{211} & 216 \\
RAIL2536 & 688.39 & $\ast$\textbf{689} & $\ast$\textbf{689} & 729 & 691 & 715 \\
RAIL2586 & 935.92 & 957 & 970 & 1012 & \textbf{947} & 988 \\
RAIL4284 & 1054.05 & 1077 & 1085 & 1148 & \textbf{1064} & 1127 \\
RAIL4872 & 1509.63 & 1549 & 1562 & 1651 & \textbf{1531} & 1591 \\ \hline
Avg. gap & & 7.53\% & 4.39\% & 55.69\% & \textbf{0.01\%} & 1.63\% \\ \hline
\end{tabular}
}
}
\end{table}

\section{Conclusion\label{sec:conclusion}}
In this paper, we considered an extension of the set covering problem (SCP) called the set multicover problem with the generalized upper bound constraints (\mbox{SMCP-GUB}).
We developed a 2-flip local search algorithm incorporated with heuristic algorithms to reduce the size of instances, evaluating promising variables by taking account of GUB constraints, and with an adaptive control mechanism of penalty weights that is used to guide the search to visit feasible and infeasible regions alternately.
We also developed an efficient implementation of a 2-flip neighborhood search that reduces the number of candidates in the neighborhood without sacrificing the solution quality.
To guide the search to visit a wide variety of good solutions, we also introduced an evolutionary approach called the path relinking method that generates new solutions by combining two or more solutions obtained so far.
According to computational comparison on benchmark instances, we can conclude that the proposed method succeeds in selecting a small number of the promising variables properly and performs quite effectively even for large-scale instances having hard GUB constraints.

We expect that the evaluation scheme of promising variables is also applicable to other combinatorial optimization problems, because the pseudo-Lagrangian score $\phi_j = \tilde{c}_j(\bm{w})$ ($j \in N$) can be defined when an adaptive control mechanism of penalty weights $\bm{w}$ is incorporated in local search, and such an approach has been observed to be effective for many problems.

\section*{Acknowledgment}
This work was supported by the Grants-in-Aid for Scientific Research (JP26282085, JP15H02969, JP15K12460).


\clearpage
\appendix
\section{Proof of Lemma~\ref{lem:nb1}\label{app:nb1}}
We show that $\Delta \hat{z}_{j_1,j_2}(\bm{x},\bm{w}) \ge 0$ holds if $x_{j_1} = x_{j_2}$.
First, we consider the case with $x_{j_1} = x_{j_2} = 1$.
By the assumption of the lemma, 
\begin{equation}
\Delta \hat{z}_j^{\downarrow}(\bm{x},\bm{w}) = - c_j + \displaystyle\sum_{i \in (M_L(\bm{x}) \cup M_E(\bm{x})) \cap S_j} w_i \ge 0
\end{equation}
holds for both $j = j_1$ and $j_2$.
Then we have 
\begin{equation}
\Delta \hat{z}_{j_1,j_2}(\bm{x},\bm{w}) = \Delta \hat{z}_{j_1}^{\downarrow}(\bm{x},\bm{w}) + \Delta \hat{z}_{j_2}^{\downarrow}(\bm{x},\bm{w}) + \displaystyle\sum_{i \in M_+(\bm{x}) \cap S_{j_1} \cap S_{j_2}} w_i \ge 0,
\end{equation}
where $M_+(\bm{x}) = \{ i \in M \mid \sum_{j \in N} a_{ij} x_j = b_i + 1 \}$.
Next, we consider the case with $x_{j_1} = x_{j_2} = 0$.
By the assumption of the lemma, for both $j= j_1$ and $j_2$,
\begin{equation}
\Delta \hat{z}_j^{\uparrow}(\bm{x},\bm{w}) = c_j - \displaystyle\sum_{i \in M_L(\bm{x}) \cap S_j} w_i \ge 0
\end{equation}
holds unless $j$ belongs a block $G_h$ satisfying $\sum_{l \in G_h} x_l = d_h$.
Then we have
\begin{equation}
\Delta \hat{z}_{j_1,j_2}(\bm{x},\bm{w}) = \Delta \hat{z}_{j_1}^{\uparrow}(\bm{x},\bm{w}) + \Delta \hat{z}_{j_2}^{\uparrow}(\bm{x},\bm{w}) + \sum_{i \in M_-(\bm{x}) \cap S_{j_1} \cap S_{j_2}} w_i \ge 0,
\end{equation}
where $M_-(\bm{x}) = \{ i \in M \mid \sum_{j \in N} a_{ij} x_j = b_i - 1 \}$.
\qed

\section{Proof of Lemma~\ref{lem:nb2}\label{app:nb2}}
We assume that neither the conditions (i) nor (ii) is satisfied and $\Delta \hat{z}_{j_1,j_2}(\bm{x},\bm{w}) < 0$ holds.
Then, for (ii) is assumed to be unsatisfied, $M_E(\bm{x}) \cap S_{j_2} \cap S_{j_2} = \emptyset$ holds, and hence we have $\Delta \hat{z}_{j_1, j_2}(\bm{x},\bm{w}) = \Delta \hat{z}_{j_1}^{\downarrow}(\bm{x},\bm{w}) + \Delta \hat{z}_{j_2}^{\uparrow}(\bm{x},\bm{w}) < 0$, which implies $\Delta \hat{z}_{j_1}^{\downarrow}(\bm{x},\bm{w}) < 0$ or $\Delta \hat{z}_{j_2}^{\uparrow}(\bm{x},\bm{w}) < 0$.
If $\Delta \hat{z}_{j_1}^{\downarrow}(\bm{x},\bm{w}) < 0$ holds, then the algorithm obtains an improved solution by flipping $x_{j_1} = 1 \to 0$.
If $\Delta \hat{z}_{j_2}^{\uparrow}(\bm{x},\bm{w}) < 0$ holds, the algorithm obtains an improved solution by flipping $x_{j_2} = 0 \to 1$, because $j_1$ and $j_2$ belong to the same block satisfying $\sum_{j \in G_h} x_j < d_h$ or to different blocks, and in the latter case, $j_2$ belongs a block $G_h$ satisfying $\sum_{j \in G_h} x_j < d_h$ by the assumption that we only consider solutions that satisfy the GUB constraints, including the solution obtained after flipping $j_1$ and $j_2$.
Both cases contradict the assumption that $\bm{x}$ is locally optimal with respect to $\mathrm{NB}_1(\bm{x})$.
\qed

\section{Proof of Lemma~\ref{lem:nb3}\label{app:nb3}}
We have $\Delta \hat{z}_{j_1,j_2}(\bm{x},\bm{w}) = \Delta \hat{z}_{j_1}^{\downarrow}(\bm{x},\bm{w}) + \Delta \hat{z}_{j_2}^{\uparrow}(\bm{x},\bm{w})$, because $M_E(\bm{x}) \cap S_{j_2} \cap S_{j_2} = \emptyset$ holds.
Then we have $\min_{j \in G_h} \Delta \hat{z}_j^{\downarrow}(\bm{x},\bm{w}) + \min_{j \in G_h} \Delta \hat{z}_j^{\uparrow}(\bm{x},\bm{w}) \le \Delta \hat{z}_{j_1}^{\downarrow}(\bm{x},\bm{w}) + \Delta \hat{z}_{j_2}^{\uparrow}(\bm{x},\bm{w}) = \Delta \hat{z}_{j_1,j_2}(\bm{x},\bm{w}) < 0$, because $\min_{j \in G_h} \Delta \hat{z}_j^{\downarrow}(\bm{x},\bm{w}) \le \Delta \hat{z}_{j_1}^{\downarrow}(\bm{x},\bm{w})$ and $\min_{j \in G_h} \Delta \hat{z}_j^{\uparrow}(\bm{x},\bm{w}) \le \Delta \hat{z}_{j_2}^{\uparrow}(\bm{x},\bm{w})$ hold.
Let $j_1^{\ast} = \mathrm{arg} \min_{j \in G_h} \Delta \hat{z}_j^{\downarrow}(\bm{x},\bm{w})$ and $j_2^{\ast} = \mathrm{arg} \min_{j \in G_h} \Delta \hat{z}_j^{\uparrow}(\bm{x},\bm{w})$.
Then we have $\Delta \hat{z}_{j_1^{\ast},j_2^{\ast}}(\bm{x},\bm{w}) = \Delta \hat{z}_{j_1^{\ast}}^{\downarrow}(\bm{x},\bm{w}) + \Delta \hat{z}_{j_2^{\ast}}^{\uparrow}(\bm{x},\bm{w}) - \sum_{i \in M_E(\bm{x}) \cap S_{j_1^{\ast}} \cap S_{j_2^{\ast}}} w_i < \Delta \hat{z}_{j_1}^{\downarrow}(\bm{x},\bm{w}) + \Delta \hat{z}_{j_2}^{\uparrow}(\bm{x},\bm{w}) < 0$.
\qed

\section{Efficient incremental evaluation of solutions in 2-FNLS\label{app:eval}}
We first consider the case that the current solution $\bm{x}$ moves to $\bm{x}^{\prime}$ by flipping $x_j = 0 \to 1$.
Then, the algorithm first updates $s_i(\bm{x})$ ($i \in S_j$) in $\mathrm{O}(|S_j|)$ time by $s_i(\bm{x}^{\prime}) \leftarrow s_i(\bm{x})+1$, and then updates $\Delta p_l^{\uparrow}(\bm{x},\bm{w})$ and $\Delta p_l^{\downarrow}(\bm{x},\bm{w})$ ($l \in N_i$, $i \in S_j$) in $\mathrm{O}(\sum_{i \in S_j} |N_i|)$ time by $\Delta p_l^{\uparrow}(\bm{x}^{\prime},\bm{w}) \leftarrow \Delta p_l^{\uparrow}(\bm{x},\bm{w}) - w_i$ if $s_i(\bm{x}^{\prime}) = b_i$ holds and $\Delta p_l^{\downarrow}(\bm{x}^{\prime},\bm{w}) \leftarrow \Delta p_l^{\downarrow}(\bm{x},\bm{w}) - w_i$ if $s_i(\bm{x}^{\prime}) = b_i+1$ holds.
Similarly, we consider the other case that the current solution $\bm{x}$ moves to $\bm{x}^{\prime}$ by flipping $x_j = 1 \to 0$.
Then, the algorithm first updates $s_i(\bm{x})$ ($i \in S_j$) in $\mathrm{O}(|S_j|)$ time by $s_i(\bm{x}^{\prime}) \leftarrow s_i(\bm{x})-1$, and then updates $\Delta p_l^{\uparrow}(\bm{x},\bm{w})$ and $\Delta p_l^{\downarrow}(\bm{x},\bm{w})$ ($l \in N_i$, $i \in S_j$) in $\mathrm{O}(\sum_{i \in S_j} |N_i|)$ time by $\Delta p_l^{\uparrow}(\bm{x}^{\prime},\bm{w}) \leftarrow \Delta p_l^{\uparrow}(\bm{x},\bm{w}) + w_i$ if $s_i(\bm{x}^{\prime}) = b_i-1$ holds and $\Delta p_l^{\downarrow}(\bm{x}^{\prime},\bm{w}) \leftarrow \Delta p_l^{\downarrow}(\bm{x},\bm{w}) + w_i$ if $s_i(\bm{x}^{\prime}) = b_i$ holds.

\end{document}